\documentclass[11pt]{article}

\usepackage{graphicx}
\usepackage{endfloat}
\usepackage{amssymb}

\usepackage{amsmath}
\usepackage{amsfonts}

\usepackage[default]{jasa_harvard}    
\usepackage{JASA_manu}

\newcommand{\cov}{\text{cov}}
\newcommand{\var}{\text{var}}
\newcommand{\Proba}{\mathbb{P}}

\newtheorem{theorem}{Theorem}[section]
\newtheorem{definition}[theorem]{Definition}
\newtheorem{corollary}[theorem]{Corollary}
\newtheorem{corrolary}[theorem]{Corollary}
\newtheorem{remark}[theorem]{Remark}
\newtheorem{proposition}[theorem]{Proposition}
\newtheorem{example}[theorem]{Example}
\newtheorem{lemma}[theorem]{Lemma}
\newtheorem{proof}[theorem]{Proof}

\begin{document}

\title{Multivariate integer-valued autoregressive\ models applied to earthquake
counts}

\author{Mathieu Boudreault\footnote{Quantact Research Group, UQAM, Research supported
by the NSERC.}  \\ 
D\'epartement de math\' ematiques \\ Universit\'e du Qu\'ebec \`a Montr\'eal \\201, avenue du Pr\'esident-Kennedy\\
Montr\'eal (Qu\'ebec)\\
Canada H2X 3Y7  \\ 
email: \texttt{boudreault.mathieu@uqam.ca}
\and
Arthur Charpentier\footnote{Quantact Research Group, UQAM, Research supported
by the Research Chair AXA/FdR ({\em Large risk in insurance}).} \\ 
D\'epartement de math\' ematiques \\ Universit\'e du Qu\'ebec \`a Montr\'eal \\201, avenue du Pr\'esident-Kennedy\\
Montr\'eal (Qu\'ebec)\\
Canada H2X 3Y7  \\ 
email: \texttt{charpentier.arthur@uqam.ca} }

\maketitle



\newpage
\begin{center}
\textbf{Abstract}
\end{center}

In various situations in the insurance industry, in finance, in
epidemiology, etc., one needs to represent the joint evolution of the
number of occurrences of an event. In this paper, we present a
multivariate integer-valued autoregressive (MINAR) model, derive its
properties and apply the model to earthquake occurrences across various
pairs of tectonic plates. The model is an extension of \cite{karlis} where cross autocorrelation (spatial contagion in a seismic
context) is considered. We fit various bivariate count models and find
that for many contiguous tectonic plates, spatial contagion is
significant in both directions. Furthermore, ignoring cross
autocorrelation can underestimate the potential for high numbers of
occurrences over the short-term. Our overall findings seem to further
confirm \cite{Parsons}.

\vspace*{.3in}

\noindent\textsc{Keywords}: {autoregressive; Granger causality; counts; earthquakes; INAR; multivariate INAR; Poisson process}
\vspace*{.3in}

\begin{center}
\textbf{Acknowledgment}
\end{center}

The authors thank Lionel Truquet for pointing out interesting references on that topic, and Fran\c{c}ois Bergeron for stimulating discussions on matrix based equations.

\newpage

\section{Introduction and motivation}

\subsection{Motivation}


Autoregression in the sense of ARIMA time series models cannot be directly
applied to integer values for obvious reasons. Thus, most integer-valued
autoregressive (INAR)\ time series models are based upon thinning operators
such as \cite{Steutel} (see also the excellent survey of
thinning operators by \cite{Weib}). Such models have been mainly
proposed and investigated by \cite{McKenzie} and \cite{AOA}
for first order autocorrelation, and by \cite{DuLi} for autocorrelation
of order $p$. \cite{Gauthier},  \cite{DGL} and \cite{Latour2} have also investigated a slightly more generalized
type of thinning operator than \cite{Steutel}, in models known
as generalized INAR (or GINAR). The statistical and actuarial literature has
multiple successful applications of INAR-type of models (see for example
\cite{Gourieroux} and \cite{BDG} where both papers treat
car insurance problems).

In a multivariate setting, the properties of a multivariate INAR\ (MINAR)
model of order 1 (based upon independent binomial thinning operators) have
been derived in \cite{franke} while the multivariate GINAR of order $%
p$ is presented in \cite{Latour1}. However, there are very few attempts in
the literature to estimate and use these types of models\footnote{%
\cite{Heinen} use a Vector Autoregression model for the mean of
two Poisson-type of random variables. Although the ultimate goal is to
represent joint integer-valued random variables, the approach taken is very
different from multivariate INAR-types of models.}. One notable exception is
\cite{karlis} and \cite{pedeli09} who investigated the bivariate INAR\ model of
order 1 with Poisson and negative binomial innovations with an application
to the number of daytime and nighttime accidents. In their papers, the
autoregression matrix is diagonal, meaning there is no cross-autocorrelation
in the counts.

Insurance policies and earthquake catastrophe (cat) derivatives (such as
cat-bonds and cat-options) offer protection against earthquake risk in
exchange for periodic premiums. Thus, one important component in these
contracts is the number of earthquakes at various locations. Earthquake
count models are mostly based upon the Poisson process (\cite{Utsu}, \cite{Gardner}, 
\cite{Lomnitz}, \cite{Kagan}), Cox process
(self-exciting, cluster or branching processes, stress-release models (see
\cite{Rathbun} for a review), or Hidden Markov Models (HMM) (see \cite{zucchini} and \cite{Orfanogiannaki}\footnote{%
For a brief summary of statistical and stochastic models in seismology, see
\cite{vere-jones}.}. However, these models are focused toward a single
location whereas seismic risk can also be influenced by shocks that occurred
at other locations (see e.g. space-time Poisson process in \cite{Ogata}, \cite{zhuang-ogata} or \cite{Schoenberg}). Thus, one of the purposes of this paper, is to propose a
bivariate INAR\ model that accounts for cross-autocorrelation in the counts.
From a seismological standpoint, that would mean the earthquake count at a
given location can be function of the past earthquake counts at that
site and at another site. These areas can be tectonic plates,
regions, cities or points on a given geological fault.

The main objective of this paper is to investigate the effects of seismic
space-time contagion (or clustering) on various risk management applications
using seismological data and a specific model. Risk management
considerations can be viewed over various time horizons. For example, prices
of cat-derivatives will be influenced by short-term earthquake risk dynamics
because a lack of an appropriate earthquake count prediction can mean
arbitrage profits or losses may occur on the markets. Insurance and
reinsurance contracts are managed over a much longer time horizon. 

\subsection{Outline of the paper}

\cite{Parsons} have confirmed that major earthquakes might have a significant impact on the number of earthquakes that occur during the hours following the main shock, but only in an area close to the main shock. They do also prove that there is no remote and large earthquakes beyond the main shock region.
Figure \ref{Fig:Nature-style-1} plots the number of quakes following a big one (magnitude exceeding 6.5), either within or outside a 2,000 km area from the main shock. One of the aims of our paper is to study the dynamics of the number of earthquakes, taking into account spatial contagion over tectonic plates. Using plates instead of distance (as in \cite{Parsons} ) allows us to work with multivariate counting processes.

\begin{figure}[ht]
\begin{center}
\includegraphics[width=0.8\textwidth]{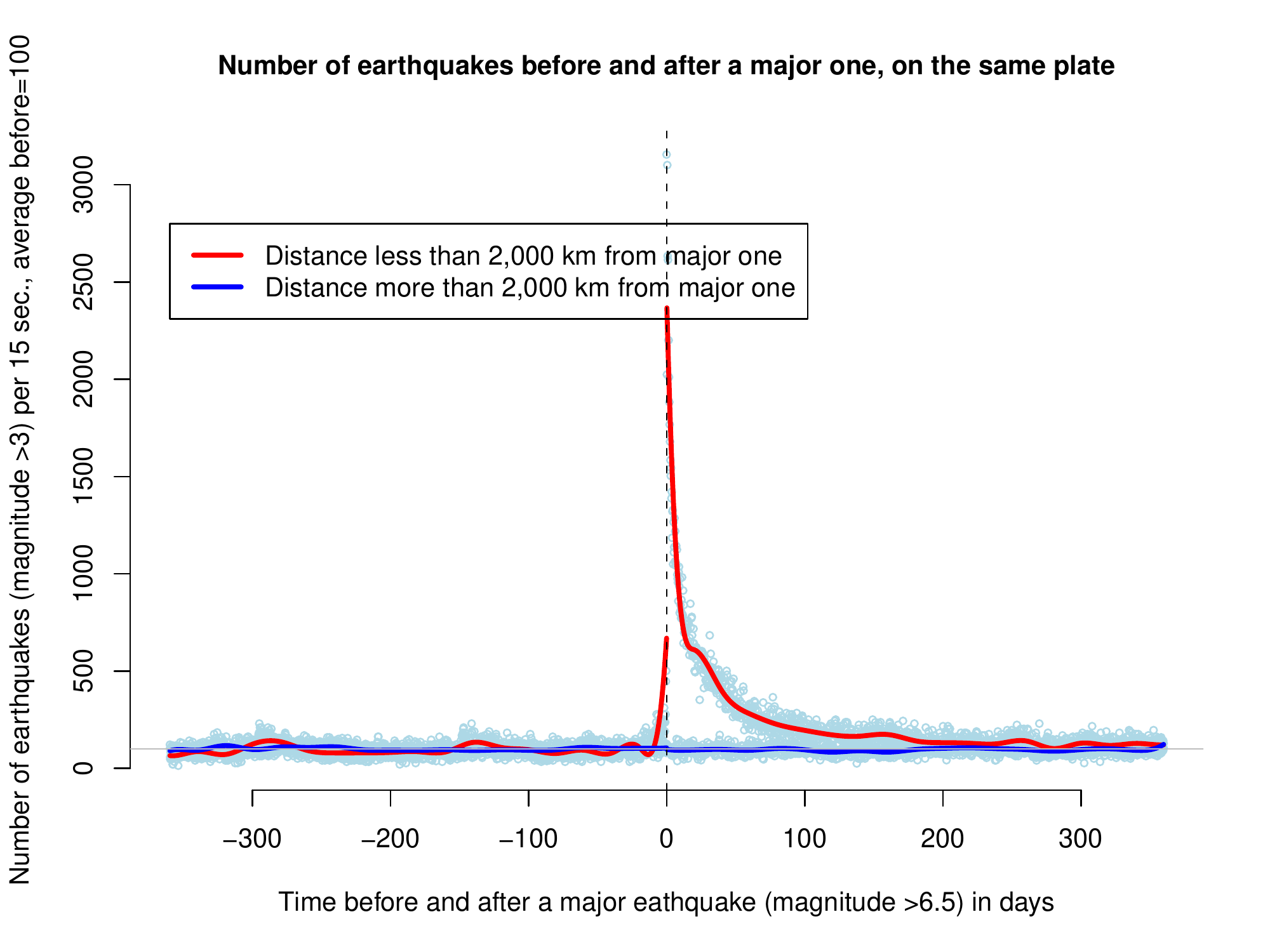}
\end{center}
\caption{Number of earthquakes (magnitude exceeding 2.0) per 15 seconds, following a large earthquake (of magnitude 6.5) either close to the main shock (less than 2,000 km) or far away (more than 2,000 km). Counts were normalized so that the expected number of earthquakes before is 100 in the two regions. Plain lines are spline regressions, either before or after the main shock).}\label{Fig:Nature-style-1}
\end{figure}

The paper is structured as follows. In Section \ref{section:MINAR}, we present the theoretical
framework of the multivariate INAR\ of order 1 and the most important
results. Some results have already  been derived in \cite{franke}
but are presented here for the sake of completeness and consistency with our given
notation. Additional (theoretical) results in terms of moments, auto-covariance functions and predictions are given in this section.
In Section \ref{section:Granger}, we introduce Granger causality tests, derived in the given context of BINAR(1) processes, including an interpretation of each coefficient in terms of causal effect.
Section \ref{section:binar-poisson} presents the specific application to the bivariate
INAR\ model with cross autocorrelation and Poisson innovations. A Monte Carlo study will also illustrate how the maximum
likelihood estimators behave (theoretical results are given in Section \ref{section:MINAR}, but only in the context of a full cross-correlation matrix).
Finally, Section \ref{section:application} provides various applications of the
model with earthquake counts. In subsection \ref{section-magn}, several BINAR(1) processes are fitted over different tectonic plates and magnitudes. We confirm here the conclusions of \cite{Parsons} claiming that the onset of a large earthquake does not cause other large ones at a very long distance. There might be contagion, but it will be between two close areas (e.g. contiguous tectonic plates), and over a short period of time (a few hours, perhaps a few days, but not much longer). 
In subsection \ref{sec:aftershock}, we have also observed that major earthquakes will generate several medium-size earthquakes on the same tectonic plate (so called aftershocks). Foreshocks were also observed, meaning that medium-size earthquakes might announce the arrival of more important earthquakes. To conclude, in subsection \ref{section:rm}, we compare the sum of counts of earthquakes on two plates, assuming that there is - or not - cross correlation between consecutive days.

\section{Multivariate integer-valued autoregression of order 1, MINAR(1)}\label{section:MINAR}

As mentioned in \cite{fokianos}, a natural way to define a linear model for counts might be to use the Poisson regression to derive an {\em autoregressive process}. Let $(N_t)$ denote a count time series, and $(\mathcal{F}_t)$ the associated filtration. A GARCH-type model can be considered, as in \cite{Ferland}
$$
N_t|\mathcal{F}_{t-1} \sim \mathcal{P}(\lambda_t), \text{ where }\lambda_t=\alpha_0+\sum_{h=1}^p\alpha_h N_{t-h} +\sum_{k=1}^q\beta_k \lambda_{t-k} .
$$
But one can easily imagine that it could be complicated (and not tractable) to extend such a process in higher dimension. An alternative can be to use a {\em thinning operator} as in \cite{AOA} or \cite{McKenzie}. The idea (introduced in \cite{Steutel}) is to define  $\circ$ as
$$
p\circ N = Y_1+\cdots+Y_N \text{ if }N\neq0, \text{ and }0\text{ otherwise},
$$
where $N$ is a random variable with values in $\mathbb{N}$, $p\in[0,1]$, and $Y_1,Y_2,\cdots$ are i.i.d. Bernoulli
variables, independent of $N$, with $\mathbb{P}(Y_i=1)=p$. Thus $p\circ N$ is a compound sum of i.i.d. Bernoulli variables.
Hence, given $N$, $p\circ N $ has a binomial distribution
with parameters $N$ and $p$. Based on that thinning operator, the integer autoregressive process of order $1$ is defined as
$$
N_t = p\circ N_{t-1}+\varepsilon_t=\sum_{i=1}^{N_{t-1}} Y_i +\varepsilon_t,
$$
where $(\varepsilon_t)$ is a sequence of i.i.d. integer valued random variables. Such  process will be called INAR(1). Note that such a process can be related to Galton-Watson process with immigration, and it is a Markov chain with integer states. As mentioned in  \cite{AOA}, if $(\varepsilon_t)$ are Poisson random variables, then $(N_t)$ will also be a sequence of Poisson random variables, and the estimation can be done easily using a method of moments estimators or maximum likelihood techniques, for $p$ and $\lambda=\mathbb{E}(\varepsilon_t)$.
One of the main interest of the thinning operator approach is that it can be easily extended in higher dimension, as in \cite{franke} (we will also provide new results, as well as new interpretations, e.g. in terms of causality).

\subsection{Thinning $\circ$ operator in dimension $d$}

As in the univariate case, before defining a multivariate counting process 
$\boldsymbol{N_t}:=(N_{1,t},\cdots,N_{d,t})$, we need to define a multivariate thinning operator for a random vector 
$\boldsymbol{N}:=(N_1,\cdots,N_d)$  with values in $\mathbb{N}^d$.
Let $\boldsymbol{P}:=[p_{i,j}]$ be a $d\times d$ matrix with entries in $[0,1]$. If 
$\boldsymbol{N}=(N_1,\cdots,N_d)$ is a random vector with values in $\mathbb{N}^d$, then
$\boldsymbol{P}\circ \boldsymbol{N}$ is a $d$-dimensional random vector, with $i$-th component
$$
[\boldsymbol{P}\circ \boldsymbol{N}]_i = \sum_{j=1}^d p_{i,j}\circ X_j,
$$
for all $i=1,\cdots,d$, where all counting variates $Y$ in $p_{i,j}\circ X_j$'s are assumed to be independent.

Note that $\boldsymbol{P}\circ (\boldsymbol{Q}\circ \boldsymbol{N}) \overset{\mathcal{L}}{=}[\boldsymbol{P}\boldsymbol{Q}]\circ \boldsymbol{N}$.
Further, from Lemma 1 in \cite{franke}, $\mathbb{E}\left(\boldsymbol{P}\circ \boldsymbol{N}\right)=\boldsymbol{P}\mathbb{E}( \boldsymbol{N})$, and
$$
\mathbb{E}\left((\boldsymbol{P}\circ \boldsymbol{N})(\boldsymbol{P}\circ \boldsymbol{N})'\right)
=\boldsymbol{P}\mathbb{E}( \boldsymbol{N} \boldsymbol{N}')\boldsymbol{P}'+\Delta,$$
with $\Delta:=\text{diag}(\boldsymbol{V}\mathbb{E}( \boldsymbol{N})
)$ where $\boldsymbol{V}$ is the $d\times d$ matrix with entries $p_{i,j}(1-p_{i,j})$.

\begin{definition}
A time series $(\boldsymbol{N}_t)$ with values in $\mathbb{N}^d$ is called a $d$-variate INAR(1) process if
\begin{equation}\label{eq:INARD-d}
\boldsymbol{N}_t = \boldsymbol{P}\circ \boldsymbol{N}_{t-1} + \boldsymbol{\varepsilon}_t
\end{equation}
for all $t$, for some $d\times d$ matrix $\boldsymbol{P}$ with entries in $[0,1]$, and some i.i.d.
random vectors $\boldsymbol{\varepsilon}_t$ with values in $\mathbb{N}^d$.
\end{definition}

\begin{remark}
\cite{karlis} and \cite{pedeli09} defined bivariate INAR(1) processes where matrix  $\boldsymbol{P}$ is a diagonal matrix. 
\end{remark}

\begin{remark}\label{remark:chain-markov}
$(\boldsymbol{N}_t)$ is a Markov chain with states in $\mathbb{N}^d$ with transition probabilities
\begin{equation}\label{eq:transition}
\pi(\boldsymbol{n}_t,\boldsymbol{n}_{t-1})=\mathbb{P}(\boldsymbol{N}_t=\boldsymbol{n}_t|
\boldsymbol{N}_{t-1}=\boldsymbol{n}_{t-1})
\end{equation}
satisfying
$$
\pi(\boldsymbol{n}_t,\boldsymbol{n}_{t-1})=\sum_{\boldsymbol{k}=0}^{\boldsymbol{n}_{t}}
\mathbb{P}(\boldsymbol{P}\circ \boldsymbol{n}_{t-1}=\boldsymbol{n}_{t}-
\boldsymbol{k})\cdot \mathbb{P}(\boldsymbol{\varepsilon}= \boldsymbol{k}).
$$
\end{remark}

\begin{remark}\label{remark:Perron-Frobenius}
Since $\boldsymbol{P}$ has entries in $[0,1]$, using a variant of Perron-Frobenius theorem, 
there exists an eigenvalue $\kappa_1$ of $\boldsymbol{P}$ such that $\kappa_1\geq |\kappa_i|$ for all other 
eigenvalues of $\boldsymbol{P}$. And the associated eigenvector 
$\boldsymbol{v}_1$ satistfies $\boldsymbol{v}_1\geq \boldsymbol{0}$.
Further, if $\boldsymbol{P}$ has strictly positive entries, $p_{i,j}\in(0,1]$, then  $\kappa_1> |\kappa_i|$
and $\boldsymbol{v}_1> \boldsymbol{0}$.
\end{remark}

From Remarks \ref{remark:chain-markov} and \ref{remark:Perron-Frobenius}, we can derive sufficient conditions so that there exists a {\em stationary} MINAR(1) process (based on Theorem 1 and Lemma 2 in \cite{franke}). The proof is based on the fact that under those assumptions,
the Markov chain is irreducible and aperiodic. One can prove that $\boldsymbol{0}$ is a positive recurrent state, and from Theorem 1.2.2 in  \cite{rosenblatt}, there exists a strictly stationary solution.

\begin{proposition}
Let $\boldsymbol{P}$ with entries in $(0,1)$, such that its largest eigenvalue is less than $1$,
and assume that $\mathbb{P}(\boldsymbol{\varepsilon}_t=
\boldsymbol{0})\in(0,1)$ with $\mathbb{E}(\|\boldsymbol{\varepsilon}_t\|_{\infty})<\infty$, then 
there exists a strictly stationary $d$-variate INAR(1) process satisfying Equation $(\ref{eq:INARD-d})$.
\end{proposition}

\begin{lemma}
Let $(\boldsymbol{N}_t)$ denote a stationary $d$-variate INAR(1) process, with autoregressive matrix $\boldsymbol{P}$ with entries in 
$(0,1)$, then $(\boldsymbol{N}_t)$ admits a $d$-variate INMA($\infty$) representation
$$
\boldsymbol{N}_t = \sum_{h=0}^{\infty}  \boldsymbol{P}^h \circ \boldsymbol{\varepsilon}_{t-h}.
$$
\end{lemma}

\subsection{Maximum likelihood estimation in $d$-variate INAR(1) processes}

Consider here a finite time series $\underline{\boldsymbol{N}}=(\boldsymbol{N}_0,
\boldsymbol{N}_1,\cdots,\boldsymbol{N}_n)$, observed from time $t=0$ until time $t=n$. The conditional
log-likelihood is
\begin{equation}\label{eq:likelihood}
\log \mathcal{L}(\underline{\boldsymbol{N}},\boldsymbol{\theta}|\boldsymbol{N}_0)
=\sum_{t=1}^n \log \pi(\boldsymbol{N}_{t-1},\boldsymbol{N}_t)
\end{equation}
where $\pi$ is the transition probability of the Markov chain, given by Equation $(\ref{eq:transition})$. Here parameter $\boldsymbol{\theta}$
is related to the autoregressive matrix $\boldsymbol{P}$ as well as parameters of the joint distribution of
the noise process, denoted $\boldsymbol{\lambda}$. For convenience, assume that $\boldsymbol{\lambda}=(\boldsymbol{\lambda}_0,\boldsymbol{\lambda}_1)$ where $\boldsymbol{\lambda}_0$ are parameters related to the innovation process $(\boldsymbol{\varepsilon}_t)$ margins, and $\boldsymbol{\lambda}_1$ to the dependence among components of the innovation.
Hence, here $\boldsymbol{\theta}=(\boldsymbol{P},\boldsymbol{\lambda})$ on some open sets $(0,1)^{d^2}\times \ell$.
From Theorem 2.2 in \cite{billingsley}, since $(\boldsymbol{N}_t)$ is a Markov chain, under standard assumptions, we can
obtain asymptotic normality of parameters.

\begin{proposition}\label{th:converge-CML}
Let $(\boldsymbol{N}_t)$ be a $d$-variate INAR(1) process satisfying stationary conditions, as well as technical assumptions (called C1-C6 in \cite{franke}),
then the conditional maximum likelihood estimate $\widehat{\boldsymbol{\theta}}$ of 
$\boldsymbol{\theta}$ is asymptotically normal, 
$$
\sqrt{n}(\widehat{\boldsymbol{\theta}}-\boldsymbol{\theta})\overset{\mathcal{L}}{\rightarrow}
\mathcal{N}(\boldsymbol{0},\Sigma^{-1}(\boldsymbol{\theta})),\text{ as }n\rightarrow\infty.
$$
Further,
$$
2[\log \mathcal{L}(\underline{\boldsymbol{N}},\widehat{\boldsymbol{\theta}}|\boldsymbol{N}_0)-
\log \mathcal{L}(\underline{\boldsymbol{N}},\boldsymbol{\theta}|\boldsymbol{N}_0)]
\overset{\mathcal{L}}{\rightarrow}
\chi^2(d^2+\text{dim}(\boldsymbol{\lambda})),\text{ as }n\rightarrow\infty.
$$
\end{proposition}

\subsection{Autocorrelation matrices for MINAR(1) processes}

Based on the properties obtained in \cite{franke} it is possible to derive expressions for autocorrelation functions, which is 
a natural way to describe the dynamics of the process.

\begin{theorem}\label{theorem:gamma}
Consider a MINAR(1) process with representation $\boldsymbol{N}_t = \boldsymbol{P}\circ \boldsymbol{N}_{t-1} + \boldsymbol{\varepsilon}_t$, where $(\boldsymbol{\varepsilon}_t)$ is the innovation process, with $\boldsymbol{\lambda}:=\mathbb{E}(\boldsymbol{\varepsilon}_{t})$ and $\boldsymbol{\Lambda}:=\var(\boldsymbol{\varepsilon}_{t})$.  Let $\boldsymbol{\mu}:=\mathbb{E}(\boldsymbol{N}_{t})$ and $\boldsymbol{\gamma}(h):=\cov(\boldsymbol{N}_{t},\boldsymbol{N}_{t-h})$.
Then $
\boldsymbol{\mu}=[\mathbb{I}-\boldsymbol{P}]^{-1} \boldsymbol{\lambda}
$ and for all $h\in\mathbb{Z}$,
$
\boldsymbol{\gamma}(h)=\boldsymbol{P}^h\boldsymbol{\gamma}(0)
$ with
$\boldsymbol{\gamma}(0)$ solution of 
$
\boldsymbol{\gamma}(0)=\boldsymbol{P}\boldsymbol{\gamma}(0)\boldsymbol{P}'+ (\boldsymbol{\Delta}+\boldsymbol{\Lambda})
$, and where $\mathbb{I}$ is the $d\times d$ identity matrix.
\end{theorem}

\begin{proof}
Since  $\mathbb{E}\left(\boldsymbol{P}\circ \boldsymbol{N}\right)=\boldsymbol{P}\mathbb{E}( \boldsymbol{N})$, then $ \boldsymbol{\mu}=\mathbb{E}(\boldsymbol{N}_{t})$ has to satisfy
$$
 \boldsymbol{\mu}=\mathbb{E}(\boldsymbol{N}_{t})=\mathbb{E}(\boldsymbol{P}\circ \boldsymbol{N}_{t-1} + \boldsymbol{\varepsilon}_t)=\boldsymbol{P} \boldsymbol{\mu}+\boldsymbol{\lambda}, \text{ i.e. }[\mathbb{I}-\boldsymbol{P}]^{-1} \boldsymbol{\lambda}.
$$
Further
$$
\boldsymbol{\gamma}(0)=\var(\boldsymbol{N}_{t})=
\mathbb{E}\left((\boldsymbol{P}\circ \boldsymbol{N}_{t-1}+\boldsymbol{\varepsilon}_{t}- \boldsymbol{\mu})(\boldsymbol{P}\circ \boldsymbol{N}_{t-1}+\boldsymbol{\varepsilon}_{t}- \boldsymbol{\mu})'\right)
$$
Since $(\boldsymbol{\varepsilon}_{t})$ is the innovation process, and from the expression mentioned above (from \cite{franke})
$$
\var(\boldsymbol{N}_{t})=
\boldsymbol{P}\var( \boldsymbol{N}_{t-1} )\boldsymbol{P}'+\boldsymbol{\Delta}+\boldsymbol{\Lambda}
$$
thus $\boldsymbol{\gamma}(0)$ satisfies
$$
\boldsymbol{\gamma}(0)=\boldsymbol{P}\boldsymbol{\gamma}(0)\boldsymbol{P}'+ (\boldsymbol{\Delta}+\boldsymbol{\Lambda}).
$$
Finally, 
$$
\boldsymbol{\gamma}(h)=\cov(\boldsymbol{N}_{t},\boldsymbol{N}_{t-h})=\cov\left((\boldsymbol{P}\circ \boldsymbol{N}_{t-1}+\boldsymbol{\varepsilon}_{t}),\boldsymbol{N}_{t-h}\right)=\cov\left((\boldsymbol{P}\circ(\boldsymbol{P}\circ \boldsymbol{N}_{t-2}+\boldsymbol{\varepsilon}_{t-1})+\boldsymbol{\varepsilon}_{t}),\boldsymbol{N}_{t-h}\right)\cdots
$$
etc, so that
$$
\boldsymbol{\gamma}(h)=\cov\left(\boldsymbol{P}^h \circ \boldsymbol{N}_{t-h}+\sum_{i=0}^{h-1}\boldsymbol{P}^i \circ \boldsymbol{\varepsilon}_{t-i},\boldsymbol{N}_{t-h}\right)=\boldsymbol{P}^h\var(\boldsymbol{N}_{t-h})=\boldsymbol{P}^h\boldsymbol{\gamma}(0),
$$
since $(\boldsymbol{\varepsilon}_t)$ is an innovation process. {\em Q.E.D.}
\end{proof}

\begin{remark}
$\boldsymbol{\gamma}(0)$ is a covariance matrix (symmetric) solution of matrix expression
$$
\boldsymbol{Z}-\boldsymbol{P}\boldsymbol{Z}\boldsymbol{P}'=\boldsymbol{A}\text{ }(=\boldsymbol{\Delta}+\boldsymbol{\Lambda}).
$$
If $\boldsymbol{P}$ was an {\em orthogonal} matrix, the term on the left could be related to Lie bracket $[\boldsymbol{Z},\boldsymbol{P}]=\boldsymbol{Z}\boldsymbol{P}-\boldsymbol{P}\boldsymbol{Z}$. If $\boldsymbol{P}$ was diagonal, we would have obtained expression of \cite{karlis}. Thus, assuming that $\boldsymbol{P}$ can either be orthogonalized or diagonalized will lead to tractable numerical algorithm. Another numerical strategy is to seek for a fixed point in equation $\boldsymbol{Z}_n=\boldsymbol{P}\boldsymbol{Z}_{n-1}\boldsymbol{P}'+\boldsymbol{A}$ with some starting value  $\boldsymbol{Z}_0$ (e.g. $\mathbb{I}$). This numerical technique will be used in the applications (see Section $\ref{section:application}$).
\end{remark}

\subsection{Forecasting with MINAR(1) processes}

In order to derive the distribution (or moments) of $\boldsymbol{N}_{t+h}$ given $\boldsymbol{N}_{t}$ recall that
$$
\left(\boldsymbol{N}_{t+h},\boldsymbol{N}_{t}\right)=
\left(\boldsymbol{P}^h \circ \boldsymbol{N}_{t}+\sum_{i=0}^{h-1}\boldsymbol{P}^i \circ \boldsymbol{\varepsilon}_{t+h-i},\boldsymbol{N}_{t-h}\right)
$$

\begin{proposition}
Let $\boldsymbol{\lambda}:=\mathbb{E}(\boldsymbol{\varepsilon}_t)$ and $\boldsymbol{\Lambda}:=\var(\boldsymbol{\varepsilon}_t)$, then
$$
\mathbb{E}(\boldsymbol{N}_{t+h}|\boldsymbol{N}_{t})=\boldsymbol{P}^h \boldsymbol{N}_{t}+\left(\mathbb{I}+\boldsymbol{P}+
\cdots +
\boldsymbol{P}^{h-1}\right)\boldsymbol{\lambda}
$$
\end{proposition}
\begin{proof}
The conditional expectation is obtained by recurrence, since the one step ahead conditional expectation is
$$
\mathbb{E}(\boldsymbol{N}_{t+1}|\boldsymbol{N}_{t})=\mathbb{E}(\boldsymbol{P} \circ\boldsymbol{N}_{t}+\boldsymbol{\varepsilon}_t |\boldsymbol{N}_{t})=\boldsymbol{P} \boldsymbol{N}_{t}+\boldsymbol{\lambda}.
$$
 {\em Q.E.D.}
\end{proof}

For the conditional variance, it is possible to derive iterative formulas, using recursions. 

\begin{proposition}
Let $\boldsymbol{\lambda}:=\mathbb{E}(\boldsymbol{\varepsilon}_t)$ and $\boldsymbol{\Lambda}:=\var(\boldsymbol{\varepsilon}_t)$, then $\var(\boldsymbol{N}_{t+h}|\boldsymbol{N}_{t})=V_h(\boldsymbol{N}_{t})$ where $V_h(\boldsymbol{N})$ is defined recursively by $V_1(\boldsymbol{N})=\text{diag}(\boldsymbol{V}\boldsymbol{N}) +\boldsymbol{\Lambda}$ and
$$
V_h(\boldsymbol{N})=\mathbb{E}[V_{h-1}(\boldsymbol{P} \circ\boldsymbol{N}+\boldsymbol{\varepsilon})|\boldsymbol{N}]+\boldsymbol{P}^{h-1} [\text{diag}(\boldsymbol{V}\boldsymbol{N}) +\boldsymbol{\Lambda}](\boldsymbol{P}^{h-1})'.
$$
\end{proposition}

\begin{proof}
The one step ahead conditional variance is
$$
\var(\boldsymbol{N}_{t+1}|\boldsymbol{N}_{t})=\var(\boldsymbol{P} \circ\boldsymbol{N}_{t}+\boldsymbol{\varepsilon}_t |\boldsymbol{N}_{t})=\text{diag}(\boldsymbol{V}\boldsymbol{N}_{t}) +\boldsymbol{\Lambda},
$$
where $\boldsymbol{V}$ is the $d\times d$ matrix with entries $p_{i,j}(1-p_{i,j})$. Then, in order to use a recursive argument, we simply have to use the variance decomposition formula, and move one additional step ahead. At time $t+h$, we can write
$$
\var(\boldsymbol{N}_{t+h}|\boldsymbol{N}_{t})=\mathbb{E}[\var(\boldsymbol{N}_{t+h}|\boldsymbol{N}_{t+1})|\boldsymbol{N}_{t}]+\var[\mathbb{E}(\boldsymbol{N}_{t+h}|\boldsymbol{N}_{t+1})|\boldsymbol{N}_{t}],
$$
i.e.
$$
\var(\boldsymbol{N}_{t+h}|\boldsymbol{N}_{t})=\mathbb{E}[\var(\boldsymbol{N}_{t+h}|\boldsymbol{N}_{t+1})|\boldsymbol{N}_{t}]+\var[\boldsymbol{P}^{h-1} \boldsymbol{N}_{t+1}+\left(\mathbb{I}+\boldsymbol{P}+
\cdots +
\boldsymbol{P}^{h-2}\right)\boldsymbol{\lambda}|\boldsymbol{N}_{t}],
$$
$$
\var(\boldsymbol{N}_{t+h}|\boldsymbol{N}_{t})=\mathbb{E}[\var(\boldsymbol{N}_{t+h}|\boldsymbol{N}_{t+1})|\boldsymbol{N}_{t}]+\boldsymbol{P}^{h-1} \var[ \boldsymbol{N}_{t+1}|\boldsymbol{N}_{t}](\boldsymbol{P}^{h-1})'.
$$
{\em Q.E.D.}
\end{proof}

In the case where $\boldsymbol{P}$ is diagonal, we obtain as particular case the expressions of \cite{pedeli09}.

\begin{corollary}
If $\boldsymbol{P}$ is a diagonal matrix, then
$$
\mathbb{E}(N_{i,t+h}|\boldsymbol{N}_{t})=p_{i,i}^h N_{i,t}+[1+p_{i,i}+\cdots+p_{i,i}^{h-1}]\lambda_i =p_{i,i}^h N_{i,t}+\left(\frac{1-p_{i,i}^h}{1-p_{i,i}}\right)\lambda_i 
$$
while
$$
\var(N_{i,t+h}|\boldsymbol{N}_{t})=p_{i,i}^h[1-p_{i,i}^h]N_{i,t}+\left(\frac{1-p_{i,i}^{2h}}{1-p_{i,i}^2}\right)\Lambda_{i,i} +\left(\frac{1-p_{i,i}^{h}}{1-p_{i,i}}-\frac{1-p_{i,i}^{2h}}{1-p_{i,i}^2}\right)\lambda_i 
$$
\end{corollary}

\begin{remark}
If $(\boldsymbol{N}_t)$ is a stationary MINAR(1) process, then, as $h\to\infty$,  $\mathbb{E}(\boldsymbol{N}_{t+h}|\boldsymbol{N}_{t})$ converges to $\sum_{k\ge 0}\boldsymbol{P}^k\boldsymbol{\mu}=[\mathbb{I}-\boldsymbol{P}]^{-1}\boldsymbol{\mu}$.
\end{remark}

\section{Nondiagonal thinning matrices and Granger causality}\label{section:Granger}

Based on the concepts introduced in \cite{franke}, it is possible to get interpretations of parameters. 
Based on Granger terminology, $N_2$ causes $N_1$ at time $t$  if and only if 
\begin{equation*}
\mathbb{E}\left( N_{1,t}|\underline{N}_{1,t-1},\underline{N}_{2,t-1}\right) \neq 
\mathbb{E}\left( N_{1,t}|\underline{N}_{1,t-1}\right),
\end{equation*}
where $\underline{N}_{1,t-1}=(N_{1,0},\cdots,N_{1,t-1})$ and $\underline{N}_{2,t-1}=(N_{2,0},\cdots,N_{2,t-1})$.

Further, $N_2$ causes {\em instantaneously} $N_1$ at time $t$ if
\begin{equation*}
\mathbb{E}\left( N_{1,t}|\underline{N}_{1,t-1},\underline{N}_{2,t-1},N_{2,t}\right) \neq 
\mathbb{E}\left( N_{1,t}|\underline{N}_{1,t-1},\underline{N}_{2,t-1}\right) .
\end{equation*}
Thus, as for Gaussian VAR processes, the following interpretation holds (see Section 3.6. in \cite{lup})
\begin{lemma}
Let $\boldsymbol{N}_t=( N_{1,t}, N_{2,t})$ be a bivariate INAR(1) with representation $(\ref{eq:INARD-d})$
\begin{enumerate}
\item $(N_{1,t})$ and $(N_{2,t})$ are instantaneously related if $\boldsymbol{\varepsilon}$ is  a noncorrelated noise, 
 \item $(N_{1,t})$ and $(N_{2,t})$ are independent, which we denote $(N_{1,t}) \perp (N_{2,t})$, if $\boldsymbol{P}$ is diagonal,
i.e. $p_{1,2}=p_{2,1}=0$, and $\varepsilon_{1,t}$ and $\varepsilon_{2,t}$ are independent,
 \item $(N_{1,t})$ causes $(N_{2,t})$ but $(N_{2,t})$ does not cause $(N_{1,t})$, which we denote $(N_{1,t}) \rightarrow
 (N_{2,t})$, if $\boldsymbol{P}$ is a lower triangle matrix,
i.e. $p_{2,1}=0$ while $p_{1,2}\neq0$,
 \item $(N_{2,t})$ causes $(N_{1,t})$ but $(N_{1,t})$ does not cause $(N_{2,t})$, which we denote $(N_{1,t}) \leftarrow
 (N_{2,t})$, if $\boldsymbol{P}$ is a lower triangle matrix,
i.e. $p_{1,2}=0$ while $p_{2,1}\neq0$,
 \item $(N_{1,t})$ causes $(N_{2,t})$ and conversely, i.e. a feedback effect, which we denote $(N_{1,t}) \leftrightarrow
 (N_{2,t})$, if $\boldsymbol{P}$ is a full matrix,
i.e. $p_{1,2},p_{2,1}\neq0$
\end{enumerate}
\end{lemma}

From those characterizations, and Proposition \ref{th:converge-CML}, it is possible to derive a simple testing procedure, based on a likelihood ratio test. 
For instantaneous causality, we test
$$
H_0:\boldsymbol{\lambda}_1=\boldsymbol{\lambda}_1^\perp \text{ against }H_1:\boldsymbol{\lambda}_1\neq \boldsymbol{\lambda}_1^\perp,
$$
where $\boldsymbol{\lambda}_1=\boldsymbol{\lambda}_1^\perp$ if and only if margins of the innovation process are independent.

\begin{corrolary}
Let $\widehat{\boldsymbol{\lambda}}$ denote the conditional maximum likelihood estimate of ${\boldsymbol{\lambda}}=({\boldsymbol{\lambda}}_0,{\boldsymbol{\lambda}}_1)$ in the non-constrained MINAR(1) model, and $\widehat{\boldsymbol{\lambda}}^\perp$ denote the conditional maximum likelihood estimate of ${\boldsymbol{\lambda}}^\perp=({\boldsymbol{\lambda}}_0,{\boldsymbol{\lambda}}_1^\perp)$ in the constrained model (when innovation has independent margins), then under suitable conditions,
$$
2[\log \mathcal{L}(\underline{\boldsymbol{N}},\widehat{\boldsymbol{\lambda}}|\boldsymbol{N}_0)-
\log \mathcal{L}(\underline{\boldsymbol{N}},\widehat{\boldsymbol{\lambda}}^\perp|\boldsymbol{N}_0)]
\overset{\mathcal{L}}{\rightarrow}
\chi^2(\text{dim}(\boldsymbol{\lambda})-\text{dim}(\boldsymbol{\lambda}^\perp)),\text{ as }n\rightarrow\infty, \text{ under }H_0.
$$
\end{corrolary}

For lagged causality, we test
$$
H_0:\boldsymbol{P}\in\mathcal{P} \text{ against }H_1:\boldsymbol{P}\notin\mathcal{P} ,
$$
where $\mathcal{P}$ is a set of constrained shaped matrix, e.g. $\mathcal{P}$ is the set of $d\times d$ diagonal matrices for lagged independence, or a set of block triangular matrices for lagged causality.
 
\begin{corrolary}
Let $\widehat{\boldsymbol{P}}$ denote the conditional maximum likelihood estimate of $\boldsymbol{P}$ in the non-constrained MINAR(1) model, and $\widehat{\boldsymbol{P}}^c$ denote the conditional maximum likelihood estimate of ${\boldsymbol{P}}$ in the constrained model, then under suitable conditions,
$$
2[\log \mathcal{L}(\underline{\boldsymbol{N}},\widehat{\boldsymbol{P}}|\boldsymbol{N}_0)-
\log \mathcal{L}(\underline{\boldsymbol{N}},\widehat{\boldsymbol{P}}^c |\boldsymbol{N}_0)]
\overset{\mathcal{L}}{\rightarrow}
\chi^2(d^2-\text{dim}(\mathcal{P})),\text{ as }n\rightarrow\infty, \text{ under }H_0.
$$
\end{corrolary}

\section{Bivariate INAR(1) process with Poisson innovation}\label{section:binar-poisson}

MINAR(d) might appear as tractable models, but the number of parameters can be extremely large. In dimension $d$, the dynamics is characterized by $d^2+\text{dim}(\boldsymbol{\Lambda})$ parameters. The standard model for the innovation process would be the multivariate common shock Poisson random vector (see \cite{Maha} or Section 37.2 in \cite{JK}). Let $(u_{I},I\subset \{1,\cdots,d\})$ be a collection of independent Poisson random variables, and define
$$
\varepsilon_i:=\sum_{I\subset \{1,\cdots,d\},i\in I} u_I
$$
then $\boldsymbol{\varepsilon}=(\varepsilon_{1},\cdots,\varepsilon_{d})$ has a multivariate Poisson distribution. In that case, $\text{dim}(\boldsymbol{\Lambda})=2^d$. In moderate dimension (e.g. $d=10$)  $\text{dim}(\boldsymbol{\Lambda})$ is larger than one thousand, which will not be tractable. Thus, for convenience, let us focus on the bivariate INAR(1) process.

\subsection{The bivariate Poisson innovation process}

A classical distribution for $\boldsymbol{\varepsilon}_t$ is the bivariate Poisson distribution, with one common shock, i.e.
$$
\left\{
\begin{matrix}
\varepsilon_{1,t} = M_{1,t}+M_{0,t}
\\
\varepsilon_{2,t} = M_{2,t}+M_{0,t}
\end{matrix}
\right.
$$
where $M_{1,t}$, $M_{2,t}$ and $M_{0,t}$ are independent Poisson variates, with parameters $\lambda_1-\varphi$, $\lambda_2-\varphi$ and $\varphi$, respectively.  
In that case, $\boldsymbol{\varepsilon}_t=(\varepsilon_{1,t},\varepsilon_{2,t})$ has joint probability function
$$
\mathbb{P}[(\varepsilon_{1,t},\varepsilon_{2,t})=(k_1,k_2)]=
e^{-[\lambda_1+\lambda_2-\varphi]}
\frac{(\lambda_1-\varphi)^{k_1}}{k_1!}
\frac{(\lambda_2-\varphi)^{k_2}}{k_2!}
\sum_{i=0}^{\min\{k_1,k_2\}}\binom{k_1}{i}\binom{k_2}{i}
i!\left(\frac{\varphi}{[\lambda_1-\varphi][\lambda_2-\varphi]}\right)
$$
with $\lambda_1,\lambda_2>0$, $\varphi\in[0, \min\{\lambda_1,\lambda_2\}]$. See e.g. 
\cite{kocherlakota}  for a comprehensive description of that joint distribution.
Note that $\varepsilon_{1,t}$ and $\varepsilon_{2,t}$ are both Poisson distributed, with parameter
$\lambda_1$ and $\lambda_2$ respectively, and here $\text{cov}(\varepsilon_{1,t},\varepsilon_{2,t})=\varphi$. Hence, parameter $\varphi$ characterizes independence (or non-independence) of the innovation process. Hence
$$
\boldsymbol{\lambda}=
\begin{pmatrix}
\lambda_1 \\ \lambda_2
\end{pmatrix} \text{ and }
\boldsymbol{\Lambda}=
\begin{pmatrix}
\lambda_1 & \varphi \\ \varphi & \lambda_2
\end{pmatrix} 
$$
and most of the previous expressions can be derived explicitly.

\subsection{BINAR(1) process with Poisson innovation}

For univariate INAR(1) processes, if $N_0$ is assumed to have a Poisson distribution with mean $\lambda/(1-p)$, then $N_t$ is also Poisson distributed, for all $t\geq 0$. But this result does not hold in higher dimensions (\cite{karlis} noticed that result with diagonal $\boldsymbol{P}$ matrices, and it is still true). Nevertheless, it is still possible to derive joint moments of the joint distributions ($\boldsymbol{\mu}$ and $\boldsymbol{\gamma}(0)$) as well as autocorrelations.

\begin{example}
$\boldsymbol{\mu}:=\mathbb{E}(\boldsymbol{N}_{t})$  is given by
$$
\left\{
\begin{matrix}
\mathbb{E}(N_{1,t})=\mu_1=\displaystyle{\frac{(1-p_{2,2})\lambda_1+p_{1,2}\lambda_2}{(1-p_{1,1})(1-p_{2,2})-p_{2,1}p_{1,2}}}\\
\mathbb{E}(N_{2,t})=\mu_2=\displaystyle{\frac{(1-p_{1,1})\lambda_2+p_{2,1}\lambda_1}{(1-p_{1,1})(1-p_{2,2})-p_{2,1}p_{1,2}}}\\
\end{matrix}
\right.
$$ 
\end{example}

Expressions for $\boldsymbol{\gamma}(0)$ and $\boldsymbol{\gamma}(1)$ can be explicitly derived, but from Theorem \ref{theorem:gamma} we do have matrices based expression that can be used numerically.

\begin{example}
Auto and cross autocorrelations are given by
$$
\text{corr}(N_{1,t},N_{2,t})=\frac{\gamma_{1,2}(0)}{\sqrt{\gamma_{1,1}(0)\gamma_{2,2}(0)}}, 
$$
$$
\text{corr}(N_{1,t},N_{1,{t-1}})=\frac{\gamma_{1,1}(1)}{\gamma_{1,1}(0)} \text{ and }\text{corr}(N_{2,t},N_{2,{t-1}})=\frac{\gamma_{2,2}(1)}{\gamma_{2,2}(0)},
$$
$$
\text{corr}(N_{1,t},N_{2,t-1})=\frac{\gamma_{1,2}(1)}{\sqrt{\gamma_{1,1}(0)\gamma_{2,2}(0)}}.
$$ 
\end{example}

Note that Poisson innovation satisfy technical assumption needed in Proposition \ref{th:converge-CML} to insure that
conditional maximum likelihood estimates converge to a normal distribution as $n$ goes to infinity.

\subsection{Maximum likelihood estimation for BINAR(1) with Poisson innovation}

From Equation \ref{eq:likelihood} the conditional likelihood of $(\boldsymbol{P},\boldsymbol{\lambda},\varphi)$ given a sample $\underline{N}=((N_{1,t},N_{2,t}),t=1,2,\cdots,n)$ is
\[
\mathcal{L}((\boldsymbol{P},\boldsymbol{\lambda},\varphi);\underline{N})=\prod_{t=1}^n
\sum_{k_1=\underline{k}_1}^{N_{1,t}}\sum_{k_2=\underline{k}_2}^{N_{2,t}}
\pi((N_{1,t}-k_{1},N_{2,t}-k_{2}),\boldsymbol{N}_{t-1}) \cdot \underbrace{\mathbb{P}[(\varepsilon_{1,t},\varepsilon_{2,t})=(k_1,k_2)]}_{\text{bivariate Poisson}}
\]
with $\underline{k}_1=\max\{N_{1,t}-N_{1,t-1}-N_{2,t-1},0\}$ and $\underline{k}_2=\max\{N_{2,t}-N_{1,t-1}-N_{2,t-1},0\}$. Here
\[
\pi((n_1,n_2),\boldsymbol{N}_{t-1})=\pi_1(n_1,\boldsymbol{N}_{t-1})\cdot \pi_2(n_2,\boldsymbol{N}_{t-1})
\]
since given with $(\varepsilon_t,\boldsymbol{N}_{t-1})$, components of $\boldsymbol{N}_{t}$ are assumed to be independent (from the definition of the multivariate thinning operator $\circ$), where $\pi_1(\cdot,\boldsymbol{N}_{t-1})$ and $\pi_2(\cdot,\boldsymbol{N}_{t-1})$ are convolutions of binomial distributions, i.e. for $n_1,n_2=0,1,\cdots,N_{1,t-1}+N_{2,t-1}$
\[
\pi_1(n_1,\boldsymbol{N}_{t-1})=\sum_{m=0}^{N_{1,t-1}} \binom{N_{1,t-1}}{m} p_{1,1}^m (1-p_{1,1})^{N_{1,t-1}-m} \binom{N_{2,t-1}}{n_1-m} p_{1,2}^{n_1-m} (1-p_{1,2})^{N_{2,t-1}-(n_1-m)} 
\]
\[
\pi_2(n_2,\boldsymbol{N}_{t-1})=\sum_{m=0}^{N_{1,t-1}} \binom{N_{1,t-1}}{m} p_{2,1}^m (1-p_{2,1})^{N_{1,t-1}-m} \binom{N_{2,t-1}}{n_2-m} p_{2,2}^{n_2-m} (1-p_{2,2})^{N_{2,t-1}-(n_2-m)} 
\]

Using numerical optimization routines, it is possible to compute $(\widehat{\boldsymbol{P}},\widehat{\boldsymbol{\lambda}},\widehat{\varphi})=\text{argmax}\{\mathcal{L}((\boldsymbol{P},\boldsymbol{\lambda},\varphi);\underline{N})\}$, and Proposition \ref{th:converge-CML} insures convergence of that estimator: the conditional maximum
likelihood estimates (CMLE)\ are asymptotically normal and unbiased.

\subsection{Monte Carlo study}

Based on the previous expression of the likelihood, it is possible to run Monte Carlo simulations to study the behavior of the estimators on simulated series.
This
numerical example illustrates with two sets of hypothetical parameters how
fast is convergence. The two sets of parameters are:

\begin{itemize}
\item $p_{1,1}=0.25,$ $p_{1,2}=0.05,$ $p_{2,1}=0.1$, $p_{2,2}=0.4$ with $%
\lambda _{1}=5$, $\lambda _{2}=3$ and $\varphi =1;$

\item $p_{1,1}=0.25,$ $p_{1,2}=p_{2,1}=0$, $p_{2,2}=0.4$ with $\lambda
_{1}=5 $, $\lambda _{2}=3$ and $\varphi =1.$
\end{itemize}

The second set of parameters is a special case of the proposed BINAR, which
is the diagonal BINAR\ model of \cite{pedeli09} and \cite{karlis}. This will
illustrate that in some instances such as $p_{1,2}=p_{2,1}=0$, the CMLE\ of
the multivariate INAR\ still converges to true values (even if Proposition \ref{th:converge-CML} insured only convergence in the interior of the support of parameters,  i.e. $(0,1)^4\times (0,\infty)^3$, not on borders).

To perform this experiment, 250 samples of different sizes have been
generated. Sample sizes of 25, 50, 100, 250, 500, 1000, 5000 and 10000
observations are considered. 


Figure \ref{Fig:mc1} shows the kernel-smoothed\footnote{%
Although, kernel smoothed densities show some curves in a negative domain,
none of estimated parameters was negative in the samples. Thus, the negative
domain is only due to smoothing.} density function of the distribution of
each parameter in the first set, over the various sample sizes. Tables \ref{Table1a} and \ref{Table1b} show the
mean and standard deviation of the parameter values over different sample
sizes. One can see that in the first set of parameters, the estimates
converge quickly to a normal distribution and the bias goes steadily to 0.
In the second set of parameters, the results are shown in Figure \ref{Fig:mc2} and Tables \ref{Table2a} and \ref{Table2b} . One sees that
even though $p_{1,2}=p_{2,1}=0$, the distribution rapidly concentrates at 0.
This indicates that the approach is valid.

\begin{figure}[ht]
\begin{center}
\includegraphics[width=0.99\textwidth]{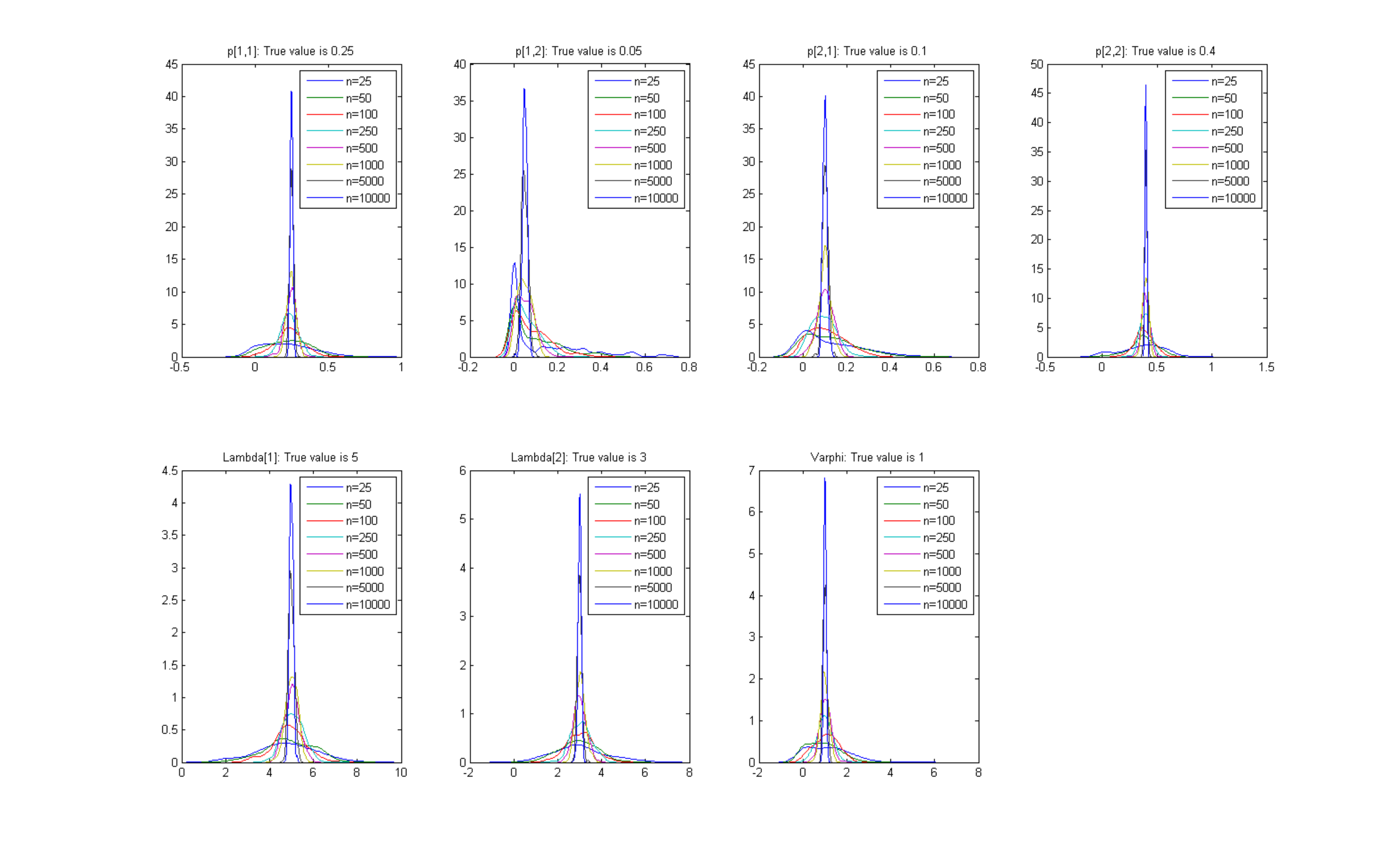}
\end{center}
\caption{Distribution of estimators $\widehat{p}_{1,1}$, $\widehat{p}_{1,2}$, $\widehat{p}_{2,1}$, $\widehat{p}_{2,2}$, $\widehat{\lambda}_{1}$, $\widehat{\lambda}_{2}$ and $\widehat{\varphi}$, as a function of the sample size $n$, case of {\em non}-diagonal $\boldsymbol{P}$ matrix.}
\label{Fig:mc1}
\end{figure}

\begin{figure}[ht]
\begin{center}
\includegraphics[width=0.99\textwidth]{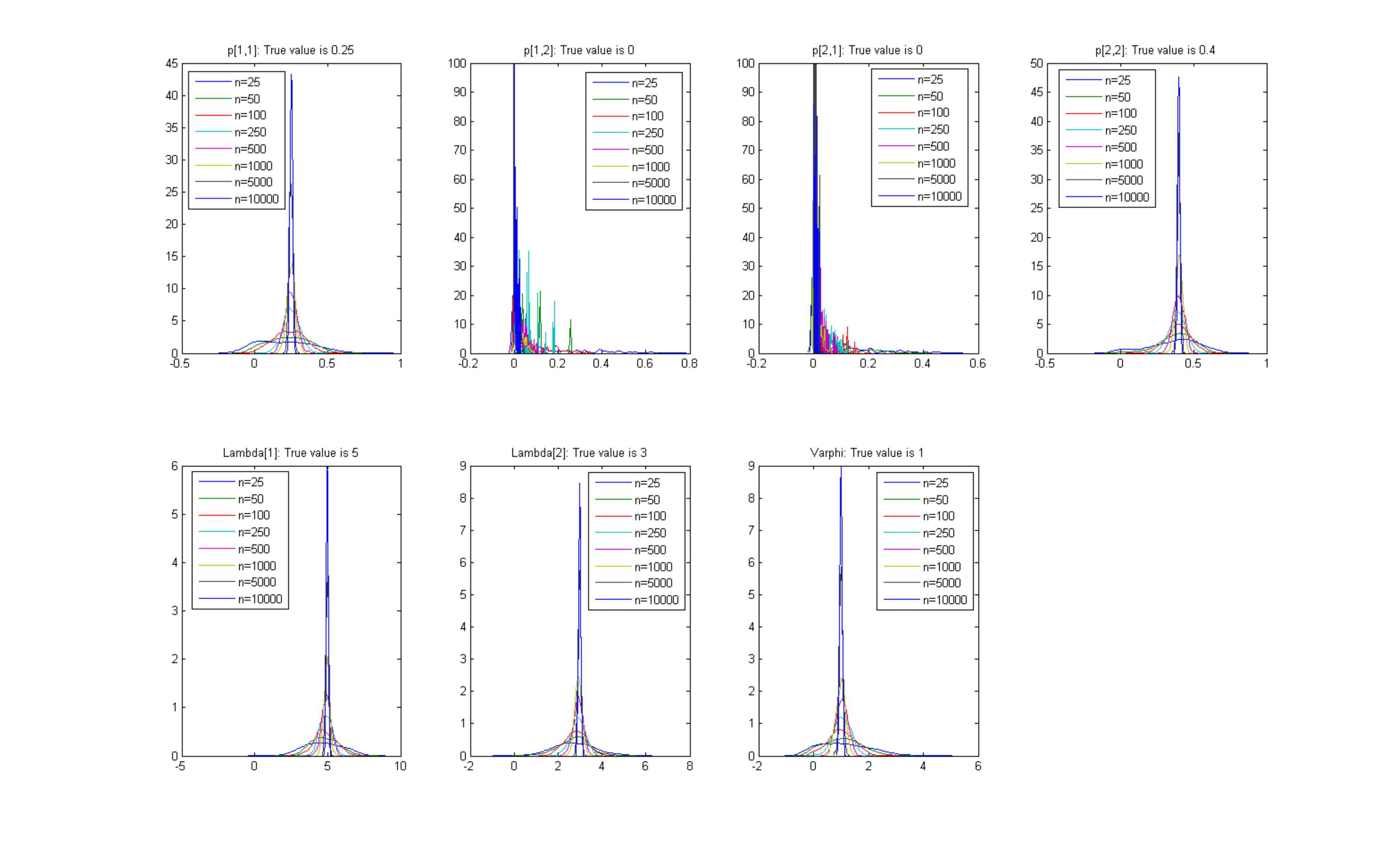}
\end{center}
\caption{Distribution of estimators $\widehat{p}_{1,1}$, $\widehat{p}_{1,2}$, $\widehat{p}_{2,1}$, $\widehat{p}_{2,2}$, $\widehat{\lambda}_{1}$, $\widehat{\lambda}_{2}$ and $\widehat{\varphi}$, as a function of the sample size $n$, case of diagonal $\boldsymbol{P}$ matrix.}
\label{Fig:mc2}
\end{figure}

\begin{table}{ 
\begin{tabular}{|r|rrrrrrr|}
\hline
Sample size  $n$ &     $\widehat{p}_{1,1}$ &     $\widehat{p}_{1,2}$ &     $\widehat{p}_{2,1}$&     $\widehat{p}_{2,2}$ &  $\widehat{\lambda}_{1}$ &  $\widehat{\lambda}_{2}$&    $\widehat\varphi$ \\
\hline
        25 &    21.27\% &    12.15\% &    14.54\% &    34.24\% &     4.8023 &     2.9854 &     1.1057 \\

        50 &    23.18\% &     9.97\% &    11.73\% &    38.62\% &     4.8538 &     2.9850 &     0.9704 \\

       100 &    23.29\% &     6.63\% &    10.11\% &    39.93\% &     5.0075 &     3.0081 &     1.0034 \\

       250 &    23.95\% &     5.52\% &    10.47\% &    39.56\% &     5.0531 &     2.9801 &     0.9774 \\

       500 &    24.57\% &     6.19\% &    10.18\% &    39.49\% &     4.9704 &     3.0318 &     1.0162 \\

      1000 &    24.93\% &     5.02\% &    10.09\% &    39.84\% &     5.0044 &     3.0040 &     0.9843 \\

      5000 &    24.88\% &     4.96\% &     9.95\% &    39.92\% &     5.0097 &     3.0086 &     0.9952 \\

     10000 &    25.06\% &     4.98\% &    10.08\% &    39.94\% &     4.9969 &     2.9972 &     1.0036 \\\hline

True value &       25\% &        5\% &       10\% &       40\% &          5 &          3 &          1 \\
\hline\end{tabular}  }
\caption{Mean parameter values - First parameter set}\label{Table1a}
\end{table}

\begin{table} 
\begin{tabular}{|r|rrrrrrr|}
 \hline
Sample size $n$ &     $\widehat{p}_{1,1}$ &     $\widehat{p}_{1,2}$ &     $\widehat{p}_{2,1}$&     $\widehat{p}_{2,2}$ &  $\widehat{\lambda}_{1}$ &  $\widehat{\lambda}_{2}$&    $\widehat\varphi$ \\
\hline

         25 &    16.87\% &    15.43\% &    14.77\% &    19.42\% &     1.2599 &     1.1232 &     0.9243 \\

        50 &    13.28\% &    11.24\% &    10.58\% &    11.98\% &     1.0497 &     0.9038 &     0.7480 \\

       100 &     9.25\% &     7.89\% &     7.99\% &     8.35\% &     0.7336 &     0.6218 &     0.5344 \\

       250 &     6.04\% &     5.24\% &     5.52\% &     5.33\% &     0.5043 &     0.4063 &     0.3701 \\

       500 &     3.94\% &     4.44\% &     3.82\% &     3.72\% &     0.3601 &     0.3106 &     0.2516 \\

      1000 &     2.94\% &     3.22\% &     2.74\% &     2.55\% &     0.2587 &     0.2144 &     0.1813 \\

      5000 &     1.37\% &     1.54\% &     1.17\% &     1.19\% &     0.1148 &     0.1002 &     0.0766 \\

     10000 &     0.92\% &     1.00\% &     0.83\% &     0.84\% &     0.0841 &     0.0660 &     0.0568 \\ \hline
\end{tabular}  
\caption{Standard deviation of parameter values - First parameter set}\label{Table1b}
\end{table}

\begin{table}{
\begin{tabular}{|r|rrrrrrr|} \hline


Sample size $n$ &     $\widehat{p}_{1,1}$ &     $\widehat{p}_{1,2}$ &     $\widehat{p}_{2,1}$&     $\widehat{p}_{2,2}$ &  $\widehat{\lambda}_{1}$ &  $\widehat{\lambda}_{2}$&    $\widehat\varphi$ \\
\hline
        25 &    20.87\% &    11.18\% &     8.15\% &    36.00\% &     4.6948 &     2.7182 &     1.1375 \\

        50 &    23.09\% &     6.10\% &     5.29\% &    38.16\% &     4.8105 &     2.7352 &     1.0503 \\

       100 &    23.22\% &     5.47\% &     3.20\% &    39.66\% &     4.8545 &     2.7997 &     1.0095 \\

       250 &    25.08\% &     2.59\% &     2.09\% &    39.24\% &     4.8667 &     2.9008 &     1.0442 \\

       500 &    24.76\% &     1.98\% &     1.38\% &    39.91\% &     4.9131 &     2.9150 &     1.0333 \\

      1000 &    24.93\% &     1.42\% &     1.00\% &    40.22\% &     4.9382 &     2.9211 &     0.9906 \\

      5000 &    24.84\% &     0.72\% &     0.38\% &    40.00\% &     4.9740 &     2.9743 &     1.0017 \\

     10000 &    25.05\% &     0.44\% &     0.34\% &    39.92\% &     4.9738 &     2.9820 &     1.0018 \\\hline

True value &       25\% &        0\% &        0\% &       40\% &          5 &          3 &          1 \\\hline

\end{tabular}  }

\caption{Mean parameter values - Second parameter set}\label{Table2a}
\end{table}

\begin{table}
{ 
\begin{tabular}{|r|rrrrrrr|}
 \hline

Sample size $n$ &     $\widehat{p}_{1,1}$ &     $\widehat{p}_{1,2}$ &     $\widehat{p}_{2,1}$&     $\widehat{p}_{2,2}$ &  $\widehat{\lambda}_{1}$ &  $\widehat{\lambda}_{2}$&    $\widehat\varphi$ \\
\hline
       25 &    17.48\% &    16.80\% &    11.65\% &    17.46\% &     1.2515 &     0.9555 &     0.8950 \\

        50 &    13.42\% &    11.42\% &     7.75\% &    12.09\% &     0.9161 &     0.7196 &     0.6621 \\

       100 &    10.05\% &     7.83\% &     4.83\% &     8.07\% &     0.7286 &     0.4765 &     0.4751 \\

       250 &     5.71\% &     4.19\% &     3.08\% &     5.37\% &     0.4263 &     0.3293 &     0.3125 \\

       500 &     4.07\% &     2.73\% &     2.04\% &     3.62\% &     0.3101 &     0.2191 &     0.2048 \\

      1000 &     2.82\% &     2.00\% &     1.36\% &     2.48\% &     0.1981 &     0.1605 &     0.1624 \\

      5000 &     1.33\% &     0.97\% &     0.61\% &     1.14\% &     0.0992 &     0.0702 &     0.0696 \\

     10000 &     0.88\% &     0.68\% &     0.53\% &     0.78\% &     0.0635 &     0.0519 &     0.0446 \\ \hline

\end{tabular}  }
\caption{Standard deviation parameter values - Second parameter set}\label{Table2b}
\end{table}

\section{Empirical application to earthquakes}\label{section:application}

\subsection{Data}

To illustrate the potential of the model for various uses, we apply the
proposed BINAR approach on earthquake counts of the Earth's tectonic plates.
Since the proposed model accounts for serial correlation and
cross-autocorrelation, earthquake counts is an interesting application for
the following reasons. First, when a mainshock occurs, it provokes many
aftershocks, thus creating serial correlation. Moreover, the seismic waves
travel over a large distance, and may cross different tectonic plates,
provoking other earthquakes on these other plates (within some time range).
This is why earthquake counts on contiguous plates should show statistical
dependence and cross autocorrelation. Given the purpose of the paper, the
empirical application is by no means an exhaustive seismological analysis of
earthquake risk across the planet. From a \emph{seismological standpoint},
some of the results are indicative and further investigation would be
required in some aspects of the application.

 The data used in the example comes from two sources. First, the limits of
each tectonic plate come from the Department of Geography of the University of Colorado at Boulder, who provide on their website, various shapefiles for
use with ArcGIS\footnote{\texttt{http://www.colorado.edu/geography/foote/maps/assign/hotspots/hotspots.html }}. Figure \ref{Fig:carte} shows the mapping of the tectonic plates in the latter reference.
The tectonic plates are:\ North American, Eurasian, Okhotsk, Pacific (split
in two, East and West), Amur, Indo-Australian, African, Indo-Chinese,
Arabian, Philippine, Coca, Caribbean, Somali, South American, Nasca and
Antarctic. We have decided not to group together the West and East Pacific
plates to keep the integrity of the input. Secondly, the listings of past
earthquakes come from the Advanced National Seismic System (ANSS) Composite
Earthquake Catalog. Each entry provides the date, time, longitude and
latitude, depth and magnitude (and its type) of each earthquake. The
database spans the time period from 1898 to 2011, but as mentioned on the
ANSS website, many databases have been added between 1898 to the mid-1960s.
Other factors may have affected the data as well. The addition of seismic
stations and the technological improvement of seismological instruments may
inflate the number of small earthquakes in the database. To twart this
issue, we focus on earthquakes with a magnitude of at least 5 ($M\geq 5$),
and we used a subset of the data. To find the most appropriate cutoff date
in the data, we used statistical tests of changes in structures (F test (Chow
test), from \cite{Andrews} or \cite{Zeilis} for implementation issues). Based on these tests, events between January 1st, 1965
up to March 30th, 2011 have been kept in the sample. The total number of
earthquakes is approximately 70~000. Finally, for $M\geq 6$ earthquakes, the
time period considered is January 1st, 1992 up to March 30th, 2011, which
amount to approximately 3000 events.

The proposed BINAR model will be largely used to investigate first-order
autocorrelation and cross-autocorrelation in earthquake counts, which is
equivalent to measuring the degree of the first order type of space and time
contagion. To do this, earthquake counts have been computed at several
frequencies. Time ranges of 3, 6, 12, 24, 36 and 48 hours have been
considered to count the number of earthquakes. 

\begin{figure}
\begin{center}
\includegraphics[width=0.8\textwidth]{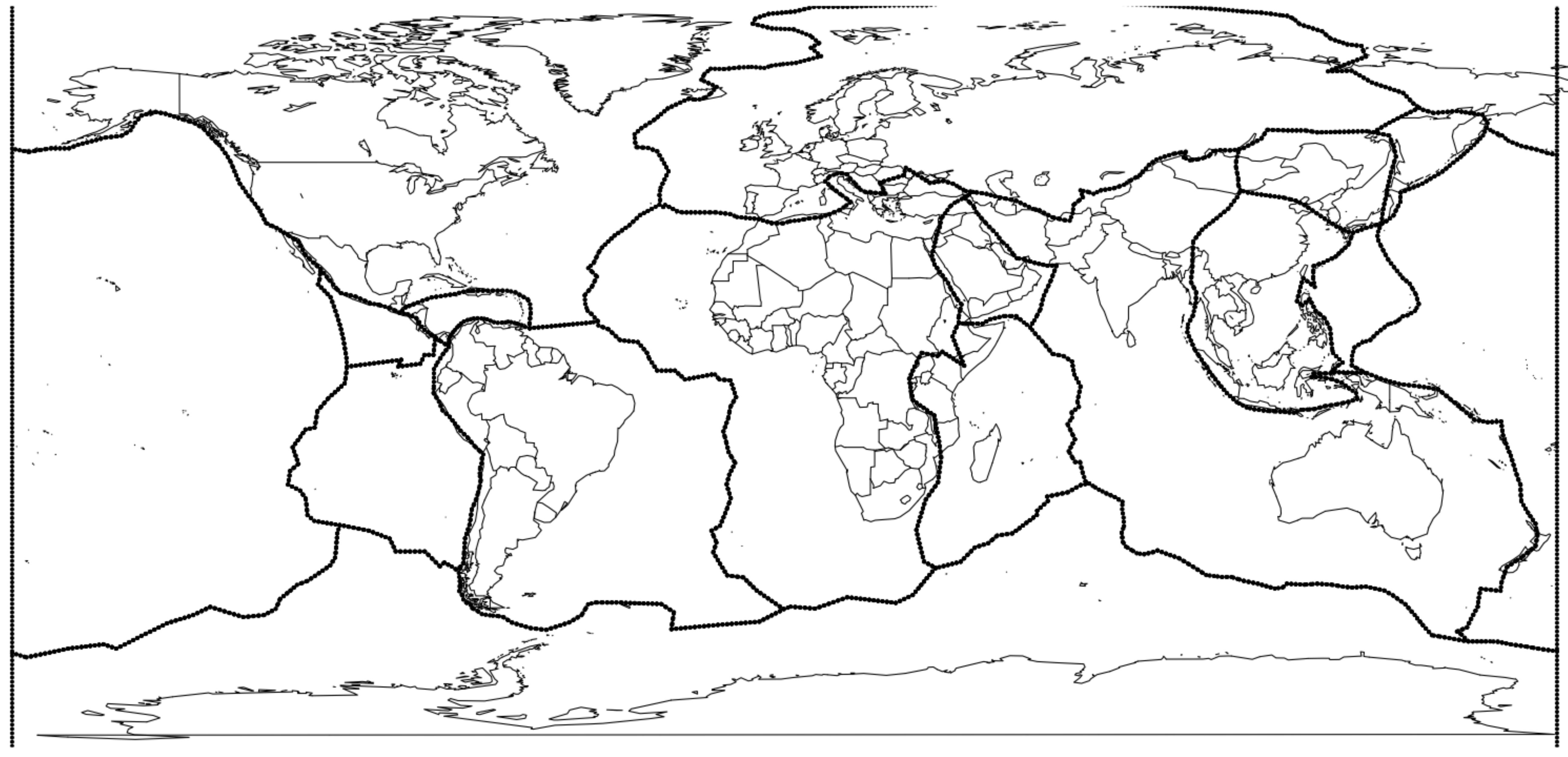}
\end{center}
\caption{The 17 tectonic plates (North American, Eurasian, Okhotsk, Pacific (split
in two, East and West), Amur, Indo-Australian, African, Indo-Chinese,
Arabian, Philippine, Coca, Caribbean, Somali, South American, Nasca and
Antarctic).}\label{Fig:carte}
\end{figure}

\subsection{Quality of fit}

The proposed BINAR model encompasses various models as well. When, $%
p_{i,j}=0,i,j=1,2$ and $\varphi =0$, there is no serial correlation, no
cross autocorrelation, and both series are independent. Those are two
independent Poisson noises. When $\varphi \not=0$, the Poisson noises are
dependent. When $p_{1,2}=p_{2,1}=0$ and $\varphi =0$, there is serial
correlation but both series are independent. Those are equivalent to two
univariate INAR processes. Finally, when $p_{1,2}=p_{2,1}=0$ and $\varphi
\not=0$, we find the diagonal BINAR\ model of \cite{karlis}.
In the latter model, there is no cross autocorrelation. Thus, in this
section we compare the fit of those four models, along with the proposed
BINAR approach, on each of the 136 possible pairs of tectonic plates.

Table \ref{Table3}
shows the results of the likelihood ratio test (LRT) of 2 dependent Poisson
noises and of 2 independent INARs, both compared with 2 independent Poisson
noises. Each column shows descriptive statistics for various sampling
frequencies. The meaning of each row is as follows:\ mean LRT across the 136
combinations of tectonic plates, standard deviation of LRTs, along with
various quantiles (50\%, 75\%, 90\%, 95\%, 97.5\%) and the proportion of
combinations that is statistically significant. Thus, the value of the LRT
provides an idea of how useful the added feature really is. The upper part
of Table \ref{Table3}
shows that dependence in the noises is important for about 10\% of the
combinations of plates (most of them are contiguous). However, serial
correlation is a much more important feature, even though noises are
independent. This should have been expected given earthquake mechanics.
Thus, the LRT\ shows that autocorrelation is one important component when
analyzing earthquake counts at these sampling frequencies.

Sampling frequency also influences the degree of autocorrelation and cross
autocorrelation. When the sampling frequency is $h$ hours, all earthquakes
on two tectonic plates will count towards $\varphi $ in the time interval $%
\left[ 0,h\right] $ hours. All earthquakes occurring in the time interval $%
\left[ h,2h\right] $ will help find first degree autocorrelation and cross
autocorrelation i.e. it should appear in the $\boldsymbol{P}$ matrix. Finally,
all earthquakes occurring after $2h$ hours, would be accounted for if a
second or third degree BINAR\ was considered. Thus, when $h$ increases, $%
\rho $ should become more significant. This is what we observe in the upper
panel of Table \ref{Table3}. Even though the percentage of combinations that are significant
very slightly increase, the mean LRT and higher percentiles of LRT tend to
grow more importantly. Note that only a few combinations are such that
tectonic plates are contiguous.

Table \ref{Table4}
shows similar computations for the diagonal BINAR\ model over the
independent INAR model (contribution of $\varphi $), and for the proposed
BINAR\ model over the diagonal BINAR\ model (contribution of $p_{1,2}$ and $%
p_{2,1}$). One sees that 6-13\% of combinations of tectonic plates show a
significant dependence in the noise (upper part of Table \ref{Table4}). Those are in large part
contiguous plates, which explains the rather low percentage. For other
combinations, plates are too far apart and independent INAR models are often
sufficient. When we compare the proposed BINAR\ model to the diagonal model
(lower part of Table \ref{Table4}), we find that 6-9\% of the combination of plates show
cross autocorrelation. In other words, the non-diagonal terms in the $%
\boldsymbol{P}$ matrix, i.e. $p_{1,2}$ and $p_{2,1}$, are both statistically
different from zero. In most of the cases, the combination of plates that
had a significant fit to both models, are contiguous. We further investigate
some of those in the next subsection. 

\subsection{Analysis of pairs of tectonic plates}\label{section-magn}

In this section, we take a look at specific pairs of tectonic plates to
observe parameters and interpret them. We will look at four different
combinations of plates, which are all closely related to Japan. The Japanese
area is one of the most seismically active regions of the world, being at
the limit of 4 tectonic plates. Table \ref{Table5} shows the (CMLE)\ parameter estimates at four
different sampling frequencies, for the four combinations of plates. The 8th
line of each panel displays the LRT over the diagonal BINAR model. A value
of more than 5.99 is significant at a level of 95\%, meaning that both cross
autocorrelation terms are significant. The last two lines show the
unconditional mean number of earthquakes per period on each plate. 

The first and second order moments estimators for the Okhotsk and West Pacific plates are given in Table \ref{Table5bis}.

Let us focus on the Okhotsk and West Pacific plates, where the results are
shown at the bottom left part of Table \ref{Table5}, and assume the sampling frequency is 24
hours. Thus, the daily number of earthquakes on the Okhotsk plate is
explained by three sources:\ the number of earthquakes on the previous day
on both plates, and a random noise effect. When no earthquake was observed
on both plates on a given day, a Poisson r.v. with mean 0.16 earthquake per
day will generate seismicity on the next day. The probability of observing
one or more earthquakes by noise only is 15\% in this case.

The interest of the proposed BINAR model lies in the representation of the
spatial contagion effect between tectonic plates. Suppose that $n$
earthquakes were observed on the Okhotsk plate on a given day, while $m$
earthquakes were observed on the West Pacific plate on that same day. The
number of earthquakes on the Okhotsk plate the next day will be the result
of the convolution of a binomial($n$, 0.0817) (autocorrelation of order 1),
a binomial($m$, 0.028) (cross autocorrelation of order 1) and a Poisson
noise (mean 0.162). Thus, on average, the number of earthquakes on the next
day on the Okhotsk plate will be%
\[
0.0817n+0.028m+0.162.
\]%
Under a diagonal model estimated with CMLE (but not provided in Table \ref{Table5}), that quantity is
\[
0.0922n+0.1748
\]%
which ignores the contribution of the West Pacific plate's earthquakes. With
$n\geq 0$ and $m\geq 1$, the diagonal model will understate the potential
number of earthquakes on the Okhotsk plate. When $m$ is relatively large,
which is one of the cases we are interested in risk management, that
understatement can be important. For example, if 3 earthquakes are observed
on the West Pacific plate, and only one on the Okhotsk plate, then we have
an average of 0.3277 earthquake on the latter plate with the proposed BINAR
model, compared with 0.267 with the diagonal model. 

Further, suppose a case where $n=0
$ and $m\geq 1$, in a context where we want to compare the mean number of
earthquakes on the Okhotsk plate using the diagonal and proposed models. In
that situation, the mean number of earthquakes in the proposed model is
roughly 16\% larger\footnote{$\frac{0.028m+0.162}{0.1748}-1$ is
approximately $\frac{0.028m}{0.1748}=0.16m.$} than the mean number of
earthquakes in the diagonal model, for each additional earthquake we observe
on the West Pacific plate (i.e. 16\% times $m$). Picking other sets of
plates, even if the LRT is much smaller and still significant, will lead to
similar analyses (Okhotsk and Amur at 12- or 24-hour frequency is one
example among others).

One may be tempted to directly compare values of $p_{1,2}$ and $p_{2,1}$ and
conclude that earthquakes on one plate provokes earthquakes on the other.
However, the gross values of $p_{1,2}$ and $p_{2,1}$ are obviously
influenced by the number of earthquakes on each plate. For example, at the
24-hour frequency, $p_{1,2}=0.028$ and $p_{2,1}=0.106$ so that one may
mistakenly pretend that Okhotsk earthquakes generally provoke earthquakes on
the West Pacific plate, and not the converse. However, there are
approximately 3 times more earthquakes on the West Pacific plate than on the
Okhotsk plate, meaning that $p_{1,2}$ has to be lower to compensate for the
higher counts on the second plate. Thus, if one is interested in determining
if earthquake counts on one plate determine the other, one should perform
Granger causality tests.

\begin{table}
\begin{tabular}{|r|rrrrrr|}

\multicolumn{ 7}{c}{{\bf Likelihood ratio test - Dependent Poisson over independent Poisson}} \\ \hline

    {\bf } &    3 hours &    6 hours &   12 hours &   24 hours &   36 hours &   48 hours \\\hline

       Mean &     6.9755 &     8.9167 &     9.7735 &     9.9896 &    10.6302 &    10.4286 \\

     Stdev &    67.8029 &    76.2997 &    83.3619 &    85.4670 &    91.0044 &    88.9381 \\

      50\% &     0.0000 &     0.0000 &     0.0000 &     0.0000 &     0.0000 &     0.0087 \\

      75\% &     0.0315 &     0.4181 &     0.5951 &     1.0178 &     1.0533 &     1.2933 \\

      90\% &     1.6329 &     3.6041 &     4.3060 &     3.7918 &     4.1652 &     4.7956 \\

      95\% &     5.1601 &     6.6625 &    11.2303 &    11.7174 &    12.8974 &    10.5802 \\

    97.5\% &     9.9646 &    16.9966 &    18.8452 &    22.7211 &    24.8961 &    19.5306 \\

 \% $>$ 3.84 &     7.35\% &    10.29\% &    12.50\% &    10.29\% &    11.03\% &    12.50\% \\ \hline

  \multicolumn{ 7}{c}{{\bf Likelihood ratio test - independent INARs over independent Poisson}} \\ \hline

    {\bf } &    3 hours &    6 hours &   12 hours &   24 hours &   36 hours &   48 hours \\\hline

        Mean &    1215.16 &    1150.81 &    1036.95 &     864.04 &     557.15 &     542.15 \\

     Stdev &    1084.90 &    1082.05 &     964.69 &     798.20 &     483.23 &     456.70 \\

      50\% &     851.01 &     781.32 &     735.33 &     561.29 &     399.63 &     416.24 \\

      75\% &    1630.64 &    1497.45 &    1415.05 &    1173.66 &     819.89 &     745.15 \\

      90\% &    2979.68 &    3030.71 &    2678.53 &    2147.04 &    1423.80 &    1319.45 \\

      95\% &    3227.93 &    3170.56 &    2837.07 &    2308.66 &    1515.82 &    1435.19 \\

    97.5\% &    3551.20 &    3589.28 &    3213.06 &    2580.33 &    1761.71 &    1650.87 \\

 \% $>$ 5.99 &   100.00\% &   100.00\% &   100.00\% &   100.00\% &   100.00\% &   100.00\% \\\hline

\end{tabular}  
\caption{Likelihood ratio test  (1) independent Poisson vector (with $\lambda\neq 0$) over independent Poisson variables (2) two independent INAR processes versus two independent Poisson variables}\label{Table3}
\end{table}

\begin{table}
\begin{tabular}{|r|rrrrrr|} 

     \multicolumn{ 7}{c}{{\bf Likelihood ratio test - diagonal BINAR over independent INARs}} \\ \hline

    {\bf } &    3 hours &    6 hours &   12 hours &   24 hours &   36 hours &   48 hours \\\hline

        Mean &     3.3744 &     4.8765 &     5.4843 &     6.3690 &     9.1950 &     8.0845 \\

     Stdev &    26.9264 &    36.7864 &    42.7106 &    52.6144 &    72.6791 &    63.8500 \\

      50\% &     0.0048 &     0.0391 &     0.0955 &     0.0171 &     0.2013 &     0.0337 \\

      75\% &     0.5303 &     0.5109 &     0.7107 &     0.8735 &     1.1663 &     1.0702 \\

      90\% &     1.8276 &     2.7837 &     4.4992 &     5.0022 &     4.2161 &     4.3279 \\

      95\% &     4.9423 &     4.7359 &     8.3814 &     9.9875 &    10.7675 &     8.4055 \\

    97.5\% &     9.2657 &    15.0495 &    13.3836 &    13.4121 &    24.5204 &    16.5784 \\

 \% $>$ 3.84 &     6.62\% &     8.82\% &    12.50\% &    13.24\% &    10.83\% &    10.83\% \\ \hline

           \multicolumn{ 7}{c}{{\bf Likelihood ratio test - proposed BINAR over diagonal BINAR}} \\ \hline

    {\bf } &    3 hours &    6 hours &   12 hours &   24 hours &   36 hours &   48 hours \\\hline

        Mean &     4.9409 &     4.1814 &     3.7077 &     4.0927 &     2.8335 &     3.5242 \\

     Stdev &    30.4265 &    25.5655 &    23.8860 &    22.6459 &    12.6545 &    15.1581 \\

      50\% &     0.9720 &     0.5379 &     0.3533 &     0.4504 &     0.3631 &     0.5931 \\

      75\% &     2.6904 &     2.3514 &     1.6824 &     2.4026 &     2.0943 &     2.3357 \\

      90\% &     5.2349 &     5.2194 &     4.0567 &     4.7533 &     4.4141 &     5.3640 \\

      95\% &     9.9654 &     7.8503 &     6.4292 &     8.5268 &     6.5271 &     8.7984 \\

    97.5\% &    15.7839 &    15.6885 &    12.0481 &    11.9986 &    11.9423 &    18.2802 \\

 \% $>$ 5.99 &     8.09\% &     8.82\% &     7.35\% &     8.82\% &     6.67\% &     9.02\% \\ \hline

\end{tabular}  
\caption{Likelihood ratio test (1) diagonal BINAR over two independent INAR processes (2) proposed BINAR over diagonal BINAR, with Poisson innovation.}\label{Table4}
\end{table}

\begin{table}
\begin{tabular}{|r|rrrr|rrrr|}
\hline
{\bf Plates} & \multicolumn{ 4}{|c|}{{\bf Okhotsk (\#1) vs. Philippine (\#2)}} & \multicolumn{ 4}{c|}{{\bf Okhotsk (\#1) vs. Amur (\#2)}} \\
\hline
{\bf Params/Frequency} &    3 hours &   12 hours &   24 hours &   48 hours &    3 hours &   12 hours &   24 hours &   48 hours \\
\hline
   $ \widehat{p}_{1,1}$ &     7.44\% &     9.45\% &    10.38\% &    12.81\% &     7.44\% &     9.44\% &    10.30\% &    12.75\% \\

    $ \widehat{p}_{1,2}$&     0.61\% &     0.60\% &     1.15\% &     0.00\% &     0.35\% &     0.83\% &     3.06\% &     2.31\% \\

     $ \widehat{p}_{2,1}$&     0.00\% &     0.00\% &     0.00\% &     0.00\% &     0.16\% &     0.42\% &     0.44\% &     0.40\% \\

    $ \widehat{p}_{2,2}$&     3.87\% &     5.83\% &     8.52\% &     8.80\% &     4.68\% &     6.44\% &     8.67\% &    10.59\% \\

  $ \widehat{\lambda}_{1}$ &     0.0222 &     0.0868 &     0.1711 &     0.3358 &     0.0223 &     0.0871 &     0.1720 &     0.3348 \\

$ \widehat{\lambda}_{2}$ &     0.0156 &     0.0612 &     0.1187 &     0.2368 &     0.0032 &     0.0122 &     0.0237 &     0.0466 \\

    $ \widehat{\varphi}$ &     0.0000 &     0.0001 &     0.0000 &     0.0021 &     0.0000 &     0.0003 &     0.0009 &     0.0024 \\

LRT (over diag.) &     4.1106 &     1.6737 &     3.6011 &     0.0000 &     2.8874 &     9.4113 &     9.5405 &     4.0631 \\

Uncond. mean (\#1) &     0.0241 &     0.0963 &     0.1926 &     0.3852 &     0.0241 &     0.0963 &     0.1927 &     0.3852 \\

Uncond. mean (\#2) &     0.0162 &     0.0650 &     0.1298 &     0.2596 &     0.0034 &     0.0134 &     0.0269 &     0.0538 \\
\hline
{\bf Plates} & \multicolumn{ 4}{|c|}{{\bf Okhotsk (\#1) vs. West Pacific (\#2)}} & \multicolumn{ 4}{c|}{{\bf Okhotsk (\#1) vs. IndoChinese (\#2)}} \\
\hline
{\bf Params/Frequency} &    3 hours &   12 hours &   24 hours &   48 hours &    3 hours &   12 hours &   24 hours &   48 hours \\
\hline
     $ \widehat{p}_{1,1}$ &     6.12\% &     7.18\% &     8.17\% &    10.13\% &     7.45\% &     9.46\% &    10.36\% &    12.83\% \\

       $ \widehat{p}_{1,2}$&     1.85\% &     2.85\% &     2.80\% &     3.13\% &     0.02\% &     0.28\% &     0.24\% &     0.10\% \\

       $ \widehat{p}_{2,1}$ &     5.84\% &     7.56\% &    10.60\% &     9.74\% &     0.22\% &     0.40\% &     0.00\% &     0.75\% \\

       $ \widehat{p}_{2,2}$ &    10.71\% &    13.52\% &    15.52\% &    15.67\% &     6.71\% &    10.29\% &    11.58\% &    13.68\% \\

  $ \widehat{\lambda}_{1}$ &     0.0214 &     0.0818 &     0.1620 &     0.3132 &     0.0223 &     0.0863 &     0.1710 &     0.3344 \\

  $ \widehat{\lambda}_{2}$ &     0.0576 &     0.2212 &     0.4261 &     0.8539 &     0.0767 &     0.2948 &     0.5818 &     1.1326 \\

   $ \widehat{\varphi}$   &     0.0012 &     0.0098 &     0.0269 &     0.0739 &     0.0002 &     0.0015 &     0.0046 &     0.0099 \\

LRT (over diag.) &   352.5998 &   275.2342 &   257.0215 &   157.0995 &     0.2839 &     3.0150 &     1.2208 &     0.6136 \\

Uncond. mean (\#1) &     0.0241 &     0.0963 &     0.1926 &     0.3852 &     0.0241 &     0.0963 &     0.1926 &     0.3852 \\

Uncond. mean (\#2) &     0.0661 &     0.2643 &     0.5285 &     1.0570 &     0.0823 &     0.3290 &     0.6580 &     1.3155 \\
\hline
\end{tabular}   
\caption{Estimation of parameters for counts of earthquakes on several tectonic plates,
Okhotsk vs. Philippine; Okhotsk vs. Amur; Okhotsk vs. West Pacific; and Okhotsk vs. Indo-Chinese plates. Includes a likelihood ratio test (null: $\boldsymbol{P}$ is a diagonal matrix). }\label{Table5}
\end{table}

\begin{table}{
\begin{tabular}{|r|rrrr|}
\hline
{\bf Plates} & \multicolumn{ 4}{c|}{{\bf Okhotsk (\#1) vs. West Pacific (\#2)}}  \\
\hline
{\bf Params/Frequency} &    3 hours &   12 hours &   24 hours &   48 hours \\
\hline
$\mathbb{E}(N_{1,t})$ &0.024 &0.096 &0.192& 0.385 \\
$\mathbb{E}(N_{2,t})$ &0.065 &0.264 &0.528 &1.057 \\
$\text{var}(N_{1,t})$     & 0.022 &0.084 &0.167 &0.326 \\
$\text{var}(N_{2,t})$     &0.060 &0.239 &0.466 &0.934 \\
$\text{cor}(N_{1,t},N_{2,t})$     & 0.038 &0.079 &0.110 &0.150 \\
$\text{cor}(N_{1,t},N_{1,t-1})$     & 0.062& 0.075 &0.086 &0.109 \\
$\text{cor}(N_{2,t},N_{2,t-1})$     & 0.108 &0.138 &0.162 &0.165 \\
$\text{cor}(N_{1,t},N_{2,t-1})$     & 0.033 &0.053 &0.055 &0.068\\
\hline
\end{tabular}   }
\caption{First and second order moments, $\boldsymbol{\mu}$ and $\boldsymbol{\gamma}(0)$ and cross-lagged correlations $\boldsymbol{\rho}(0)$ and $\boldsymbol{\rho}(1)$, for counts on two plates, Okhotsk vs. West Pacific.}\label{Table5bis}
\end{table}

\subsection{Foreshocks and aftershocks}\label{sec:aftershock}

As another application of the model, we illustrate the relationship between
medium-size earthquakes (i.e. $5\leq M\leq 6$) and large earthquakes ($M>6$%
). Using the proposed BINAR model in this context will help understand how
the size of a set of earthquakes at a given time period can help predict the
size of future earthquakes. Most of the time, large earthquakes (mainshocks)
are followed by aftershocks, which are usually smaller (medium-size or
smaller). The inverse, in which case a medium-size earthquake may announce a
larger earthquake, is usually less likely, but still regularly observed.
Figure \ref{Fig:Nature-style-2} illustrates this relationship between foreshocks, mainshocks
and aftershocks.

\begin{figure}[ht]
\begin{center}
\includegraphics[width=0.8\textwidth]{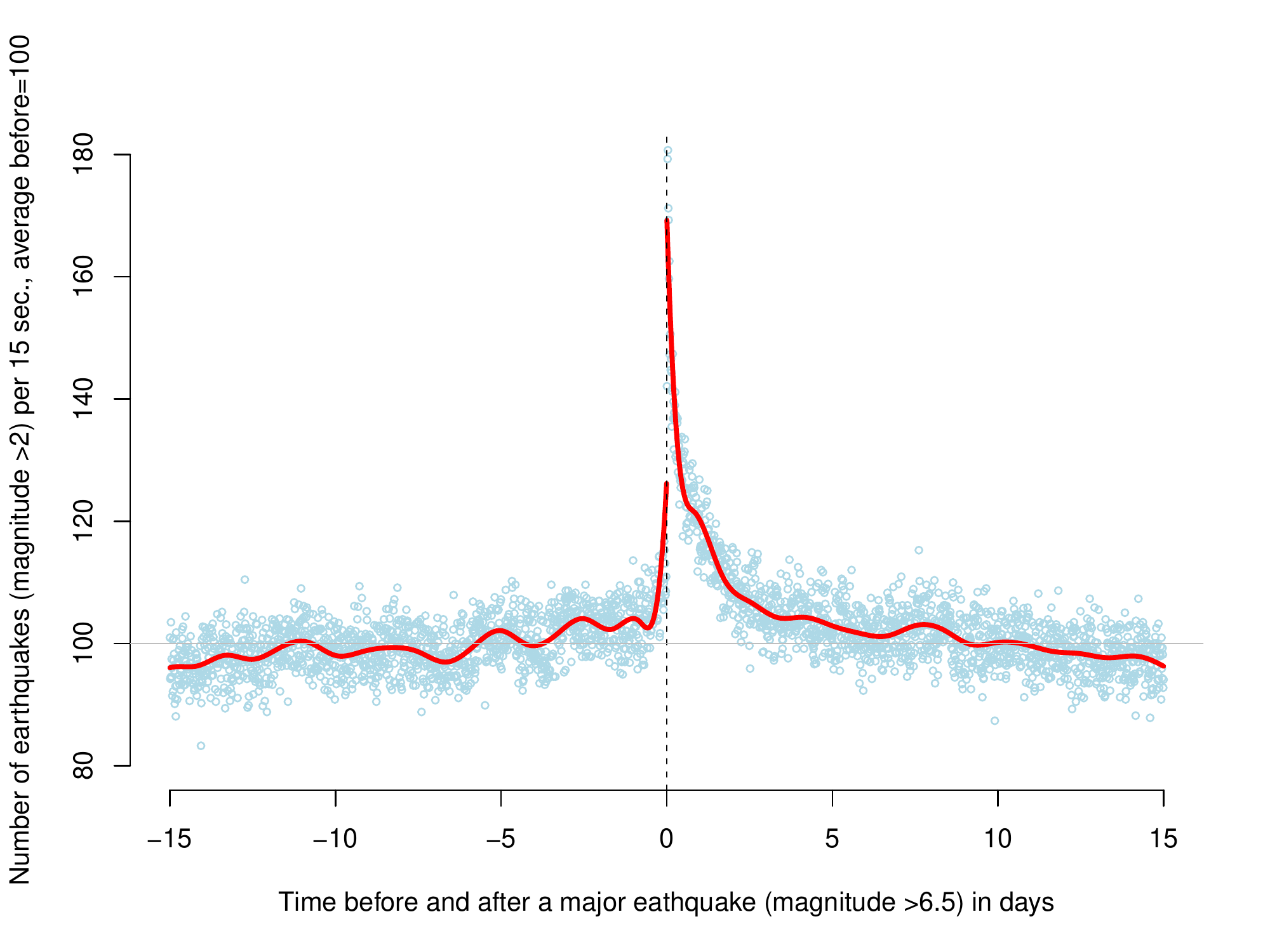}
\end{center}
\caption{Number of earthquakes (magnitude exceeding 2.0) per 15 seconds, following a large earthquake (of magnitude 6.5), normalized so that the expected number of earthquakes before is 100. Plain lines are spline regressions, either before or after the main shock).}\label{Fig:Nature-style-2}
\end{figure}

As a first exercise, we have fitted the same five models, that is the
proposed BINAR model, the diagonal BINAR model, independent INARs, dependent
Poisson noises, and independent Poisson noises. According to the LRT, the
fit of the diagonal BINAR\ model over independent Poisson noises is
statistically significant for all tectonic plates, at all sampling
frequencies. This is also the case when the diagonal BINAR is compared to
independent INAR models. Finally, for all but a few tectonic plates and/or
sampling frequencies, the diagonal BINAR model has a very significant fit
over dependent Poisson noises.

Thus, for this application, we would like to measure if cross
autocorrelation is important, i.e. if earthquake size on a given period
helps explain future earthquake sizes. Table \ref{Table6} shows the LRT for the proposed BINAR
model over the diagonal model, for various sampling frequencies. A value
larger than 5.99 means that $p_{1,2}\not=0$ and $p_{2,1}\not=0$, implying
that large earthquakes are followed by medium-size earthquakes, and the
opposite also holds. This is indeed the case in the large majority of
tectonic plates, although this relationship clearly gets weaker when the
sampling frequency goes from 3-hours to 48-hours (last row of the table).
This should have been expected given Omori's law, which explains the
temporal decay of aftershock rates.

Table \ref{Table7}
shows the CMLE\ parameter estimates for all plates at the 12-hour frequency.
The last two columns provide the unconditional mean number of medium ($5\leq
M\leq 6$) and large ($M>6$) earthquakes. With a 12-hour sampling frequency,
only the Coca and Somali plates have an unsignificant LRT at a level of
95\%, meaning that for 15 plates, cross autocorrelation is important. Thus,
one should not directly compare values of $p_{1,2}$ and $p_{2,1}$ since only
Granger causality tests will provide the true significance of contagion
between the two sets of data.

Let us illustrate the impact of cross autocorrelation for a given tectonic
plate. Assume that on the Okhotsk plate, which seats beneath part of Japan,
there is a large earthquake in the prior 12-hour period (and no medium-size
earthquake). Then, cross autocorrelation will be the most important
component of the mean number of earthquakes in the next period. Indeed, the
expected number of medium earthquakes in the next period is $%
0.2444+0.0780=0.3224$ and cross autocorrelation will account for more than
the two thirds of the total expectation. One can compare the size of $p_{1,2}
$ and $p_{2,1}$ with the noise components ($\lambda $s) and observe that the
ratio is much larger in this section than in Section \ref{section-magn}. Thus,
cross sectional effects are a key element in this context. Finally, the
ratio of the expected number of $M>6$ earthquakes over $5\leq M\leq 6$
earthquakes is on average (across plates)\ approximately 10 which is
consistent with Gutenberg and Richter's law. 


\begin{table}
\begin{tabular}{|r|rrrrrr|}
\hline
{\bf Plate name} & {\bf 3 hours} & {\bf 6 hours} & {\bf 12 hours} & {\bf 24 hours} & {\bf 36 hours} & {\bf 48 hours} \\
\hline
North American &    60.9470 &    31.3554 &    21.6005 &    16.5618 &    15.9287 &     6.0150 \\

  Eurasian &     3.4172 &     2.4860 &    17.5732 &     1.3990 &     8.4201 &     1.0743 \\

   Okhotsk &   135.5948 &   109.1666 &   109.3060 &   113.5049 &    36.3703 &    52.2677 \\

East Pacific &    37.2827 &    50.2991 &    32.2566 &    19.9613 &    19.1339 &     4.4437 \\

West Pacific &   101.3846 &    96.3865 &   110.2205 &   109.0303 &    62.4744 &    81.3029 \\

      Amur &    12.9162 &    17.2652 &     7.7396 &     4.0498 &    10.2767 &    15.0012 \\

Indo-Australian &   303.2257 &   233.9429 &   169.1037 &   124.9187 &    75.0355 &    48.7183 \\

   African &    35.0197 &    11.1661 &    15.9194 &    12.7146 &    28.1233 &     9.0930 \\

Indo-Chinese &    63.2515 &    29.9391 &    49.5970 &    64.0781 &    29.8289 &    45.1555 \\

   Arabian &     4.5921 &     4.5763 &    12.3768 &     3.1358 &     2.1765 &     0.1744 \\

Philippine &     9.2969 &    21.3144 &    20.1805 &    15.5858 &    18.9310 &    17.5329 \\

      Coca &    12.6070 &    15.1147 &     3.1709 &     8.2198 &     5.3469 &     9.0246 \\

 Caribbean &    20.4764 &    24.5509 &    21.2256 &     3.4367 &     7.6112 &     2.4771 \\

    Somali &     0.2432 &     5.1726 &     3.2162 &     0.0039 &     0.1625 &     0.0392 \\

South American &    81.8145 &    58.0135 &    50.4781 &    55.8060 &    77.3867 &    58.4621 \\

     Nasca &    76.8393 &    38.6514 &    20.2549 &    17.6382 &     8.6659 &    11.1903 \\

 Antarctic &     2.9275 &     9.2410 &     9.1584 &     4.7911 &     1.8339 &     0.9290 \\
\hline
Average LRT &    56.5786 &    44.6260 &    39.6105 &    33.8139 &    23.9827 &    21.3471 \\
\hline
\end{tabular}   
\caption{Likelihood Ratio Test for the proposed BINAR model over the diagonal model, for various sampling frequencies, when $N_{1,t}$ denotes the number of medium size earthquakes (magnitude between $5$ and $6$) during period $t$, and $N_{1,t}$ denotes the number of large earthquakes (magnitude exceeding $6$) during period $t$.}\label{Table6}
\end{table}

\begin{table}{\small
\begin{tabular}{|r|rrrrrrrrrr|}
\hline
{\bf Plate name} & {\bf $\widehat{p}_{1,1}$} & {\bf $\widehat{p}_{1,2}$} & {\bf $\widehat{p}_{2,1}$} & {\bf $\widehat{p}_{2,2}$} & {\bf $\widehat{\lambda}_{1}$} & {\bf $\widehat{\lambda}_{2}$} & {\bf $\widehat{\varphi}$} & {\bf $\widehat{p}_{1,2}/\widehat{p}_{2,1}$} & {\bf Mean$<$6} & {\bf Mean$>$6} \\
\hline
North American &    0.05633 &    0.11372 &    0.00444 &    0.01027 &     0.0844 &     0.0091 &     0.0028 &      25.63 &     0.0906 &     0.0096 \\

  Eurasian &    0.01348 &    0.14431 &    0.01082 &    0.00006 &     0.0181 &     0.0011 &     0.0002 &      13.34 &     0.0185 &     0.0013 \\

   Okhotsk &    0.11224 &    0.24445 &    0.00995 &    0.01951 &     0.0780 &     0.0104 &     0.0033 &      24.57 &     0.0910 &     0.0115 \\

East Pacific &    0.07950 &    0.12959 &    0.00385 &    0.00025 &     0.2631 &     0.0285 &     0.0075 &      33.64 &     0.2900 &     0.0296 \\

West Pacific &    0.15688 &    0.21797 &    0.00642 &    0.01163 &     0.1995 &     0.0212 &     0.0084 &      33.93 &     0.2426 &     0.0231 \\

      Amur &    0.00931 &    0.09470 &    0.00676 &    0.02041 &     0.0107 &     0.0024 &     0.0008 &      14.02 &     0.0110 &     0.0025 \\

Indo-Australian &    0.19749 &    0.24562 &    0.01079 &    0.03095 &     0.4039 &     0.0490 &     0.0225 &      22.76 &     0.5205 &     0.0564 \\

   African &    0.03906 &    0.13683 &    0.00204 &    0.00716 &     0.0564 &     0.0054 &     0.0014 &      67.20 &     0.0595 &     0.0056 \\

Indo-Chinese &    0.09744 &    0.16198 &    0.00563 &    0.00956 &     0.2501 &     0.0236 &     0.0080 &      28.76 &     0.2816 &     0.0254 \\

   Arabian &    0.04026 &    0.24457 &    0.00347 &    0.00009 &     0.0167 &     0.0007 &     0.0001 &      70.51 &     0.0176 &     0.0008 \\

Philippine &    0.03630 &    0.09681 &    0.00864 &    0.03512 &     0.0536 &     0.0055 &     0.0012 &      11.20 &     0.0563 &     0.0062 \\

      Coca &    0.06228 &    0.04115 &    0.00155 &    0.00534 &     0.0439 &     0.0069 &     0.0020 &      26.61 &     0.0471 &     0.0070 \\

 Caribbean &    0.03080 &    0.26310 &    0.00001 &    0.00009 &     0.0083 &     0.0008 &     0.0004 &   
31969 &     0.0088 &     0.0008 \\

    Somali &    0.02325 &    0.02809 &    0.00384 &    0.00000 &     0.0284 &     0.0012 &     0.0001 &       7.32 &     0.0291 &     0.0013 \\

South American &    0.13661 &    0.12043 &    0.01141 &    0.01507 &     0.1384 &     0.0160 &     0.0046 &      10.55 &     0.1628 &     0.0181 \\

     Nasca &    0.11426 &    0.13361 &    0.00307 &    0.01442 &     0.0378 &     0.0034 &     0.0013 &      43.49 &     0.0433 &     0.0036 \\

 Antarctic &    0.02875 &    0.03897 &    0.00879 &    0.00153 &     0.0548 &     0.0056 &     0.0010 &       4.43 &     0.0567 &     0.0061 \\
\hline
\end{tabular}   }
\caption{CMLE estimators for the proposed BINAR model, for 12-hour frequency, when $N_{1,t}$ denotes the number of medium size earthquakes (magnitude between $5$ and $6$) during period $t$, and $N_{1,t}$ denotes the number of large earthquakes (magnitude exceeding $6$) during period $t$.}\label{Table7}
\end{table}

\subsection{Risk management}\label{section:rm}

In many risk management applications, such as the computation of premiums
and reserves or the pricing of catastrophe derivatives, the total loss
amount over a given area, region, or city is what matters most. One
important driver of the total loss amount, is the total number of
earthquakes over the area in question for various time horizons $\left[ 0,T%
\right] .$ In this section, we compare the distribution of the sum of the
number of earthquakes over a given area, for the diagonal and the proposed
BINAR models. We do so for pairs of tectonic plates (see Section \ref{section-magn}) where the LRT was
statistically significant, otherwise the two models are too similar.

Two sets of tectonic plates are analyzed: (1) Okhotsk and West Pacific
plates (Japan is at the limits of the West Pacific, Okhotsk, Philippine and
Amur plates) and (2) South American and Nasca plates (which holds the South
American continent and the West Coast of South America (Chile for example)\
is located at the limit of these two plates). Two extreme scenarios are
generated. In the first set of plate, we assume that 23 earthquakes were
observed on the Okhotsk tectonic plate and 46 were observed on the West
Pacific plate (this is indeed what happened in the last 12 hours of March
10th, 2011). In the second set of plates, we assume that 24 earthquakes were
observed on the South American plate, whereas only 3 were observed on the
Nasca plate (this is what occurred on the second half of February 27th,
2010). Using 100~000 paths of a bivariate diagonal INAR and the proposed
bivariate INAR\ models, we have computed the total number of earthquakes
that occurred on both plates (of a given set), on the next $T$ days ($%
T=1,3,7,14$ and 30). The results are shown in Table \ref{Table8}. The left (right) panel
focuses on the first (second) set of tectonic plates. The numbers shown are $%
\mathbb{P} \left( \left. \sum_{k=1}^{T}\left( N_{1,k}+N_{2,k}\right) \geq
n\right\vert \mathcal{F}_{0}\right) $ for various values of $n$.

One sees that the diagonal model really understates the number of
earthquakes in the following days, especially in the tails. For example, in
the first set of plates (Okhotsk and West Pacific), the probability of
having a total of at least 20 earthquakes in the next day is 6.7\% with the
proposed model, whereas it is 0.7\% with the diagonal model; it is a
ten-fold increase. This increase is all due to the non-diagonal terms in the
$\boldsymbol{P}$ matrix as it accounts for the cross auto-correlation between
earthquake counts. A less dramatic increase is observed in the second set of
plates (South America and Nasca). For example, the probability of having a
total of at least 7 earthquakes over a week on both plates is 39.6\% in the
diagonal model whereas this probability is 44\% in the proposed model. As
expected, over the long-term, both processes converge to their equilibrium
and the effect of the initial conditions seem to dissipate.

We now suppose that with both sets of plates, no earthquake occurred on a
given day. Table \ref{Table9} shows the results of $\Proba \left( \left. \sum_{k=1}^{T}\left(
N_{1,k}+N_{2,k}\right) \geq n\right\vert \mathcal{F}_{0}\right) $ for $T=14$
and 30 days. For smaller $T$ values, the probabilities generated by the two
models are very similar since it takes a lot of time to develop earthquakes
and thus to observe cross-sectional effects. For the given $T$ values, the
probabilities are very similar for both models, with a slightly fatter tail
for the proposed model in the first set of tectonic plates. In the second
set of plates, the probabilities are too close to be able to conclude of any
difference.

In summary, we have also run different scenarios on other sets of plates and
it confirms that the effect of the non-diagonal terms in the $\boldsymbol{P}$
matrix is to generate fatter tails in the sum of the number of earthquakes.
This is very useful for short-term risk management applications such as the
pricing of earthquake bonds and other derivatives. An underestimation of the
number of earthquakes could mean arbitrage opportunities if the market model
has a similar behavior to the diagonal model.

\begin{table}{\small
\begin{tabular}{|crrrrr|crrrrr|}
\hline
\multicolumn{ 6}{|c}{{\bf Diagonal model (Okhotsk and West Pacific)}} & \multicolumn{ 6}{|c|}{{\bf Diagonal model (South American and Nasca)}} \\
\hline
  n / days &      1 day &     3 days &     7 days &    14 days &    30 days &   n / days &      1 day &     3 days &     7 days &    14 days &    30 days \\
\hline
         5 &     0.9680 &     0.9869 &     0.9978 &     0.9999 &     1.0000 &          2 &     0.8489 &     0.9166 &     0.9757 &     0.9981 &     1.0000 \\

        10 &     0.5650 &     0.7207 &     0.8972 &     0.9884 &     0.9999 &          5 &     0.2708 &     0.4321 &     0.6965 &     0.9277 &     0.9988 \\

        15 &     0.1027 &     0.2270 &     0.4978 &     0.8548 &     0.9985 &          7 &     0.0685 &     0.1628 &     0.3959 &     0.7655 &     0.9906 \\

        20 &     0.0067 &     0.0277 &     0.1308 &     0.4997 &     0.9752 &         10 &     0.0035 &     0.0192 &     0.1041 &     0.4108 &     0.9334 \\

        25 &     0.0003 &     0.0018 &     0.0170 &     0.1684 &     0.8588 &         15 &     0.0000 &     0.0002 &     0.0033 &     0.0547 &     0.5885 \\

        30 &     0.0000 &     0.0001 &     0.0014 &     0.0319 &     0.5965 &         20 &     0.0000 &     0.0000 &     0.0000 &     0.0031 &     0.1873 \\

        40 &     0.0000 &     0.0000 &     0.0000 &     0.0002 &     0.1034 &         25 &     0.0000 &     0.0000 &     0.0000 &     0.0001 &     0.0290 \\

        50 &     0.0000 &     0.0000 &     0.0000 &     0.0000 &     0.0041 &         30 &     0.0000 &     0.0000 &     0.0000 &     0.0000 &     0.0030 \\
\hline
\multicolumn{ 6}{|c}{{\bf Proposed model  (Okhotsk and West Pacific)}} & \multicolumn{ 6}{|c|}{{\bf Proposed model  (South American and Nasca)}} \\
\hline
  n / days &      1 day &     3 days &     7 days &    14 days &    30 days &   n / days &      1 day &     3 days &     7 days &    14 days &    30 days \\
\hline
         5 &     0.9946 &     0.9977 &     0.9997 &     1.0000 &     1.0000 &          2 &     0.8780 &     0.9321 &     0.9805 &     0.9979 &     1.0000 \\

        10 &     0.8344 &     0.9064 &     0.9712 &     0.9970 &     1.0000 &          5 &     0.3323 &     0.4888 &     0.7331 &     0.9362 &     0.9990 \\

        15 &     0.3638 &     0.5288 &     0.7548 &     0.9479 &     0.9995 &          7 &     0.0990 &     0.2034 &     0.4410 &     0.7913 &     0.9921 \\

        20 &     0.0671 &     0.1573 &     0.3616 &     0.7256 &     0.9917 &         10 &     0.0082 &     0.0309 &     0.1271 &     0.4435 &     0.9386 \\

        25 &     0.0053 &     0.0246 &     0.0970 &     0.3815 &     0.9357 &         15 &     0.0000 &     0.0004 &     0.0056 &     0.0688 &     0.6145 \\

        30 &     0.0002 &     0.0023 &     0.0151 &     0.1268 &     0.7646 &         20 &     0.0000 &     0.0000 &     0.0001 &     0.0039 &     0.2099 \\

        40 &     0.0000 &     0.0000 &     0.0001 &     0.0038 &     0.2335 &         25 &     0.0000 &     0.0000 &     0.0000 &     0.0001 &     0.0380 \\

        50 &     0.0000 &     0.0000 &     0.0000 &     0.0001 &     0.0221 &         30 &     0.0000 &     0.0000 &     0.0000 &     0.0000 &     0.0036 \\
\hline
\end{tabular}  }
\caption{Empirical evolution of $%
\mathbb{P} \left( \left. \sum_{k=1}^{T}\left( N_{1,k}+N_{2,k}\right) \geq
n\right\vert \mathcal{F}_{0}\right) $ for various values of $n$ (per line) and $T$ (per column), on two plates (Okhotsk vs. West Pacific and South American vs. Nasca), either for a diagonal $\boldsymbol{P}$ matrix - on top - or for a full matrix - below.}\label{Table8}
\end{table}

\begin{table}
\begin{tabular}{|rrr|rrr|}
\hline
                                 \multicolumn{ 6}{|c|}{{\bf Diagonal model}} \\
\hline
\multicolumn{ 3}{|c}{{\bf (Okhotsk and West Pacific plates)}} & \multicolumn{ 3}{|c|}{{\bf (South American and Nasca plates)}} \\
\hline
  n / days &    14 days &    30 days &   n / days &    14 days &    30 days \\
\hline
         5 &     0.9495 &     1.0000 &          3 &     0.9096 &     0.9991 \\

        10 &     0.5281 &     0.9943 &          5 &     0.6710 &     0.9904 \\

        15 &     0.1138 &     0.9125 &          7 &     0.3733 &     0.9506 \\

        20 &     0.0121 &     0.6302 &         10 &     0.0962 &     0.7770 \\

        25 &     0.0008 &     0.2791 &         12 &     0.0296 &     0.5883 \\

        30 &     0.0000 &     0.0774 &         15 &     0.0038 &     0.2996 \\

        35 &     0.0000 &     0.0137 &         20 &     0.0001 &     0.0498 \\

        40 &     0.0000 &     0.0018 &         25 &     0.0000 &     0.0036 \\

        45 &     0.0000 &     0.0002 &         30 &     0.0000 &     0.0002 \\

        50 &     0.0000 &     0.0000 &         35 &     0.0000 &     0.0000 \\
\hline
                                 \multicolumn{ 6}{|c|}{{\bf Proposed model}} \\
\hline
\multicolumn{ 3}{|c}{{\bf (Okhotsk and West Pacific plates)}} & \multicolumn{ 3}{|c|}{{\bf (South American and Nasca plates)}} \\
\hline
  n / days &    14 days &    30 days &   n / days &    14 days &    30 days \\
\hline
         5 &     0.9444 &     0.9999 &          3 &     0.9061 &     0.9990 \\

        10 &     0.5211 &     0.9927 &          5 &     0.6662 &     0.9899 \\

        15 &     0.1261 &     0.9033 &          7 &     0.3754 &     0.9497 \\

        20 &     0.0139 &     0.6242 &         10 &     0.0990 &     0.7745 \\

        25 &     0.0007 &     0.2877 &         12 &     0.0317 &     0.5856 \\

        30 &     0.0000 &     0.0870 &         15 &     0.0049 &     0.2986 \\

        35 &     0.0000 &     0.0177 &         20 &     0.0001 &     0.0513 \\

        40 &     0.0000 &     0.0024 &         25 &     0.0000 &     0.0044 \\

        45 &     0.0000 &     0.0003 &         30 &     0.0000 &     0.0002 \\

        50 &     0.0000 &     0.0001 &         35 &     0.0000 &     0.0000 \\
\hline
\end{tabular}  
\caption{Empirical evolution of $%
\mathbb{P} \left( \left. \sum_{k=1}^{T}\left( N_{1,k}+N_{2,k}\right) \geq
n\right\vert \mathcal{F}_{0}\right) $ for various values of $n$ and two time horizon $T$, on two plates (Okhotsk vs. West Pacific and South American vs. Nasca), either for a diagonal $\boldsymbol{P}$ matrix, or for a full matrix.}\label{Table9}
\end{table}

\section{Conclusion}

In this paper, we confirm the conclusion of \cite{Parsons} claiming that very large earthquakes do not necessarily cause large ones at a very long distance. There might be contagion, but it will be within two close areas (e.g. contiguous tectonic plates), and over a short period of time (a few hours, perhaps a few days, but not much longer). Nevertheless, not taking into account possible spatial contagion between consecutive periods may lead to large underestimation of overall counts. In the context of foreshocks, mainshocks and aftershocks, we have also observed that major earthquakes might generate several medium-size earthquakes on the same tectonic plate, and also foreshocks might announce possible large earthquakes.

\bibliographystyle{ECA_jasa}
\bibliography{BINAR}

\end{document}